\documentclass[11pt]{article}
\pdfoutput=1
\usepackage[T1]{fontenc}
\usepackage{lmodern}
\usepackage{slantsc}
\usepackage[protrusion=true,expansion=true]{microtype}
\usepackage{amsmath,amssymb,amsfonts,amsthm}
\usepackage{subcaption}
\usepackage{graphicx}
\usepackage{fullpage}
\usepackage{setspace}
\usepackage[backref=page]{hyperref}
\usepackage{color}
\usepackage{wrapfig}
\usepackage{tikz}
\usetikzlibrary{decorations.pathreplacing}
\usepackage{algorithm}
\usepackage[noend]{algpseudocode}
\usepackage[framemethod=tikz]{mdframed}
\usepackage{xspace}
\usepackage{pgfplots}
\usepackage{framed}
\usepackage{enumitem}
\usepackage{thmtools}
\usepackage{thm-restate}
\usepackage{tabu}
\usepackage{fancyhdr}
\pgfplotsset{compat=1.5}

\newtheorem{theorem}{Theorem}[section]
\newtheorem{corollary}[theorem]{Corollary}
\newtheorem{lemma}[theorem]{Lemma}
\newtheorem{proposition}[theorem]{Proposition}
\newtheorem{definition}[theorem]{Definition}

\newenvironment{proofof}[1]{\begin{trivlist} \item {\bf Proof
#1:~~}}
  {\qed\end{trivlist}}

\newcommand{\namedref}[2]{\hyperref[#2]{#1~\ref*{#2}}}
\newcommand{\thmlab}[1]{\label{thm:#1}}
\newcommand{\thmref}[1]{\namedref{Theorem}{thm:#1}}
\newcommand{\lemlab}[1]{\label{lem:#1}}
\newcommand{\lemref}[1]{\namedref{Lemma}{lem:#1}}

\newcommand{\corlab}[1]{\label{cor:#1}}
\newcommand{\corref}[1]{\namedref{Corollary}{cor:#1}}
\newcommand{\seclab}[1]{\label{sec:#1}}
\newcommand{\secref}[1]{\namedref{Section}{sec:#1}}

\newcommand{\alglab}[1]{\label{alg:#1}}
\renewcommand{\algref}[1]{\namedref{Algorithm}{alg:#1}}
\newcommand{\tablelab}[1]{\label{tab:#1}}
\newcommand{\tableref}[1]{\namedref{Table}{tab:#1}}
\newcommand{\deflab}[1]{\label{def:#1}}
\newcommand{\defref}[1]{\namedref{Definition}{def:#1}}

\newcommand{\propref}[1]{\namedref{Property}{prop:#1}}
\newcommand{\proplab}[1]{\label{prop:#1}}

\def \d    {\mdef{\mathrm{d}}}

\newcommand{\CountSketch}{\textsc{CountSketch}}
\newcommand{\FpEst}{\textsc{FpEst}}
\newcommand{\AMS}{\textsc{AMS}}

\newcommand{\PPr}[1]{\ensuremath{\mathbf{Pr}\left[#1\right]}}
\newcommand{\PPPr}[2]{\ensuremath{\underset{#1}{\mathbf{Pr}}\left[#2\right]}}
\newcommand{\Ex}[1]{\ensuremath{\mathbb{E}\left[#1\right]}}
\newcommand{\EEx}[2]{\ensuremath{\underset{#1}{\mathbb{E}}\left[#2\right]}}
\renewcommand{\O}[1]{\ensuremath{\mathcal{O}\left(#1\right)}}
\newcommand{\tO}[1]{\ensuremath{\tilde{\mathcal{O}}\left(#1\right)}}
\newcommand{\eps}{\varepsilon}
\newcommand{\Var}[1]{\ensuremath{\text{Var}\left[#1\right]}}

\def \calA    {\mdef{\mathcal{A}}}
\def \calB    {\mdef{\mathcal{B}}}

\def \calD    {\mdef{\mathcal{D}}}
\def \calE    {\mdef{\mathcal{E}}}

\def \calN    {\mdef{\mathcal{N}}}

\def \calP    {\mdef{\mathcal{P}}}
\def \calQ    {\mdef{\mathcal{Q}}}

\def \calS    {\mdef{\mathcal{S}}}

\def \calU    {\mdef{\mathcal{U}}}
\def \calV    {\mdef{\mathcal{V}}}

\def \be    {\mdef{\mathbf{e}}}

\def \bu    {\mdef{\mathbf{u}}}
\def \bw    {\mdef{\mathbf{w}}}
\def \bv    {\mdef{\mathbf{v}}}

\def \bz    {\mdef{\mathbf{z}}}

\newcommand{\mdef}[1]{{\ensuremath{#1}}\xspace}  

\DeclareMathOperator*{\argmax}{argmax}
\DeclareMathOperator*{\polylog}{polylog}
\DeclareMathOperator*{\poly}{poly}
\DeclareMathOperator*{\rnd}{rnd}

\newcommand{\ignore}[1]{}

\newif\ifnotes\notestrue 
\ifnotes
\newcommand{\samson}[1]{\textcolor{blue}{{\bf (Samson:} {#1}{\bf ) }} \marginpar{\tiny\bf
             \begin{minipage}[t]{0.5in}
               \raggedright S:
            \end{minipage}}}
\newcommand{\shenghao}[1]{\textcolor{red}{{\bf (Shenghao:} {#1}{\bf ) }} \marginpar{\tiny\bf
             \begin{minipage}[t]{0.5in}
               \raggedright S:
            \end{minipage}}}
\newcommand{\david}[1]{\textcolor{purple}{{\bf (David:} {#1}{\bf ) }} \marginpar{\tiny\bf
             \begin{minipage}[t]{0.5in}
               \raggedright D:
            \end{minipage}}} 
\else
\newcommand{\samson}[1]{}
\newcommand{\shenghao}[1]{}
\newcommand{\david}[1]{}
\fi

\makeatletter
\renewcommand*{\@fnsymbol}[1]{\textcolor{mahogany}{\ensuremath{\ifcase#1\or *\or \dagger\or \ddagger\or
 \mathsection\or \triangledown\or \mathparagraph\or \|\or **\or \dagger\dagger
   \or \ddagger\ddagger \else\@ctrerr\fi}}}
\makeatother

\providecommand{\email}[1]{\href{mailto:#1}{\nolinkurl{#1}\xspace}}

\definecolor{mahogany}{rgb}{0.75, 0.25, 0.0}
\definecolor{darkblue}{rgb}{0.0, 0.0, 0.55}
\definecolor{darkpastelgreen}{rgb}{0.01, 0.75, 0.24}
\definecolor{darkgreen}{rgb}{0.0, 0.2, 0.13}
\definecolor{darkgoldenrod}{rgb}{0.72, 0.53, 0.04}
\definecolor{darkred}{rgb}{0.55, 0.0, 0.0}
\definecolor{forestgreenweb}{rgb}{0.13, 0.55, 0.13}
\definecolor{greencss}{rgb}{0.0, 0.5, 0.0}
\definecolor{bleudefrance}{rgb}{0.19, 0.55, 0.91}

\hypersetup{
     colorlinks   = true,
     citecolor    = mahogany,
	 linkcolor	  = darkpastelgreen
}

\fancypagestyle{pg}
{
\lhead{}
\rhead{}
\cfoot{--\ \thepage\ --}

}

\AtBeginDocument{%
  \DeclareFontShape{T1}{lmr}{m}{scit}{<->ssub*lmr/m/scsl}{}%
}
\begin{document}

\allowdisplaybreaks

\title{Perfect Sampling in Turnstile Streams Beyond Small Moments}
\author{
David P. Woodruff\thanks{Carnegie Mellon University and Google Research. 
E-mail: \email{dwoodruf@cs.cmu.edu}. 
}
\and
Shenghao Xie\thanks{Texas A\&M University. 
E-mail: \email{xsh1302@gmail.com}.}
\and
Samson Zhou\thanks{Texas A\&M University. 
E-mail: \email{samsonzhou@gmail.com}.}
}
\maketitle

\begin{abstract}
Given a vector $x \in \mathbb{R}^n$ induced by a turnstile stream $S$, a non-negative function $G: \mathbb{R} \to \mathbb{R}$, a perfect $G$-sampler outputs an index $i$ with probability $\frac{G(x_i)}{\sum_{j\in[n]} G(x_j)}+\frac{1}{\text{poly}(n)}$. Jayaram and Woodruff (FOCS 2018) introduced a perfect $L_p$-sampler, where $G(z)=|z|^p$, for $p\in(0,2]$. In this paper, we solve this problem for $p>2$ by a sampling-and-rejection method. Our algorithm runs in $n^{1-2/p} \cdot \text{polylog}(n)$ bits of space, which is tight up to polylogarithmic factors in $n$. Our algorithm also provides a $(1+\varepsilon)$-approximation to the sampled item $x_i$ with high probability using an additional $\varepsilon^{-2} n^{1-2/p} \cdot \text{polylog}(n)$ bits of space.

Interestingly, we show our techniques can be generalized to perfect polynomial samplers on turnstile streams, which is a class of functions that is not scale-invariant, in contrast to the existing perfect $L_p$ samplers. We also achieve perfect samplers for the logarithmic function $G(z)=\log(1+|z|)$ and the cap function $G(z)=\min(T,|z|^p)$. Finally, we give an application of our results to the problem of norm/moment estimation for a subset $\mathcal{Q}$ of coordinates of a vector, revealed only after the data stream is processed, e.g., when the set $\mathcal{Q}$ represents a range query, or the set $n\setminus\mathcal{Q}$ represents a collection of entities who wish for their information to be expunged from the dataset. 
\end{abstract}

\section{Introduction}
As databases manage increasingly larger real-time information, the streaming model of computation has become a crucial setting to analyze algorithms for processing massive, dynamic datasets, such as real-time social media feeds, sensor data for smart cities, live video analytics, detecting and preventing distributed denial of service (DDoS) attacks, and real-time indexing and querying in large-scale databases. 
In the one-pass streaming model, a frequency vector on an underlying universe $[n]$ is implicitly defined through sequential updates to the coordinates of the vector. 
These updates can only be observed once and the goal is to aggregate statistics about the frequency vector while using space that is sublinear in the size of both the frequency vector and the data stream. 

Sampling items from the dataset is a central and versatile technique for analyzing large-scale datasets. 
For example, various works have explored sampling methods in the context of big data applications~\cite{Vitter85,GemullaLH08, CohenDKLT11,CohenDKLT14}, such as virtual network traffic monitoring~\cite{GilbertKMS01, EstanV03,MaiCSYZ06,HuangNGHJJT07,ThottanLJ10}, database management~\cite{LiptonNS90,HaasS92, LiptonN95,HaasNSS96,GibbonsM98,Haas16,CohenG19}, distributed computing~\cite{TirthapuraW11,CormodeMYZ12,CormodeF14,WoodruffZ16,JayaramSTW19}, and data summarization~\cite{FriezeKV04,DeshpandeRVW06, DeshpandeV07,AggarwalDK09,IndykMMM14,MahabadiIGR19, IndykMGR20,MahabadiRWZ20,cohen2020wor}. 
Formally, there exists an underlying vector $x\in\mathbb{R}^n$, which is defined through a sequence of $m$ updates. 
For each $t\in[m]$, an update $(i_t,\Delta_t)$ changes coordinate $x_{i_t}\in[n]$ by $\Delta_t\in\{-M,\ldots,M\}$. 
Therefore, for each $i\in[n]$, the $i$-th coordinate of the frequency vector $x$ is defined by 
\[x_i=\sum_{t\in[m]: i_t=i}\Delta_t.\]
As the updates $\{\Delta_t\}$ are permitted to both increase and decrease coordinates of $x$, this is called the \emph{turnstile} model of streaming, whereas in the \emph{insertion-only} model, all updates must satisfy $\Delta_t\ge 0$. 
The goal is then to extract a coordinate $i\in[n]$, possibly along with an estimate of $x_i$, with probability proportional to $G(x_i)$ for some function $G$:
\begin{definition}[$G$-sampler]
Given $x\in\mathbb{R}^n$, $\eps\ge 0$, and a non-negative function $G:\mathbb{R}\to\mathbb{R}$, an $\eps$-approximate $G$-sampler outputs an index $i^*\in[n]$ with probability at least $\frac{2}{3}$, or otherwise it returns a failure symbol $\bot$. 
Furthermore, conditioned on $i^*\neq\bot$, we have for each $i\in[n]$:
\[\PPr{i=i^*}=\left(1\pm\eps\right)\cdot\frac{G(x_i)}{\sum_{j\in[n]}G(x_j)}+n^{-c},\]
where $c>0$ is a constant input parameter.
If $\eps=0$, we say the sampler is \emph{perfect}. 
\end{definition}

\subsection{Related Work}
\paragraph{$L_p$ samplers.} 
The most popular choice for $G(z)$ is the class of functions $G(z)=|z|^p$ for $p>0$, known as $L_p$ samplers. 
Introduced by \cite{MonemizadehW10}, $L_p$ samplers have been used to design streaming algorithms for heavy hitters, $L_p$ norm and $F_p$ moment estimation, cascaded norm approximation, and finding duplicates \cite{MonemizadehW10,AndoniKO11,JowhariST11, BravermanOZ12,JayaramW18,cohen2020wor,JayaramWZ22}. 
For insertion-only streams, the classic technique of reservoir sampling~\cite{Vitter85} acquires a truly perfect $L_1$ sample using $\O{\log n}$ bits of space, i.e., $\eps=0$ and furthermore there is no additive $\frac{1}{\poly(n)}$ distortion in the sampling probabilities. 
However, for either $p\neq 1$ or  turnstile streams, the problem becomes significantly more difficult and thus, the existence of sublinear-space $L_p$ samplers was asked by Cormode, Murthukrishnan, and Rozenbaum~\cite{CormodeMR05}. 

The question was first answered in the affirmative by Monemizadeh and Woodruff~\cite{MonemizadehW10}, who gave an $\eps$-approximate $L_p$ sampler for the turnstile model that uses $\poly\left(\frac{1}{\eps},\log n\right)$ bits of space for $p\in[1,2]$. 
These $L_p$ samplers were improved first by Andoni, Krauthgamer, and Onak~\cite{AndoniKO11} and subsequently by Jowhari, Saglam, and Tardos~\cite{JowhariST11}, to use roughly $\O{\frac{1}{\eps^{\max(1,p)}}\log^2 n}$ bits of space for $p\in(0,2)$ and $\O{\frac{1}{\eps^2}\log^3 n}$ bits of space for $p=2$. 
\cite{JowhariST11} also showed a lower bound of $\Omega(\log^2 n)$ space for $p<2$ but curiously there were no known lower bounds in terms of $\eps$. 
This gap was explained by \cite{JayaramW18}, who gave a perfect $L_p$ sampler that uses $\tO{\log^2 n}$ bits of space for $p\in(0,2)$ and $\O{\log^3 n}$ bits of space for $p=2$. 
However, it is not immediately clear how to extend their techniques to $p>2$. 
Although truly perfect $L_p$ samplers for turnstile streams were ruled out by \cite{JayaramWZ22}, truly perfect $L_p$ samplers for insertion-only streams were obtained by \cite{JayaramWZ22} using $\tO{n^{1-1/p}}$ space for $p>1$ and by \cite{PettieW25} using $\O{\log n}$ bits of space for $p\in(0,1)$, albeit in the random oracle model. 

\paragraph{Other $G$-samplers.} 
Other popular choices of $G(z)$ for $G$-samplers include the logarithmic function $G(z)=\log(1+z)$, the cap function $G(z)=\min(T,z^p)$ for a threshold $T$, and the concave sublinear functions $G(z)=\int_0^\infty a(t)\min(1,zt)\,dt$, where $a(t)\ge 0$~\cite{CohenG19}. 
However, significantly less is known about these functions. 
\cite{JayaramWZ22} gave truly perfect $G$-samplers for monotonically increasing functions $G:\mathbb{R}\to\mathbb{R}_{\ge 0}$ on insertion-only streams with $G(0)=0$, using space roughly proportional to $\frac{\|x\|_1}{G(X)}$, where $G(X)=\sum_{i\in[n]}G(x_i)$. 
In particular, their results apply to $M$-estimators such as the $L_1-L_2$ estimator $G(z)=2\left(\sqrt{1+\frac{z^2}{2}}-1\right)$, the Fair estimator $G(z)=\tau|z|-\tau^2\log\left(1+\frac{|z|}{\tau}\right)$, and the Huber estimator $G(z)=\frac{z^2}{2\tau}$ for $|z|\le\tau$ and $G(z)=|z|-\frac{\tau}{2}$ otherwise, where $\tau>0$ is some constant parameter. 
\cite{CohenG19} approximated the class of concave sublinear functions with the class of soft concave sublinear functions, i.e., $G(z)=\int_0^\infty a(t)(1-e^{-zt})\,dt$ and developed approximate $G$-samplers on insertion-only streams for soft concave sublinear functions. 
\cite{PettieW25} subsequently developed truly perfect samplers on insertion-only streams for the class of functions 
\[G(z)=c\cdot\mathbb{1}[z>0]+\gamma_0 z+\int_0^\infty(1-e^{-tz})\,\nu(dt),\]
which has a bijection with the Laplace exponents of non-negative, one-dimensional L\'{e}vy processes. 
This class of functions includes the $L_p$ samplers $G(z)=z^p$ for $p\in(0,1)$, the soft-cap sampler $G(z)=1-e^{\tau z}$, and $G(z)=\log(1+z)$. 
Their truly perfect $G$-samplers use only two words of memory, but require both the random oracle model and the insertion-only streaming model, and cannot immediately output an estimate of the sampled coordinate $x_i$. 

\subsection{Our Contributions}
In this paper, we present a number of new techniques and applications for sampling from turnstile streams. 
\paragraph{Perfect $L_p$ samplers.} 
We present the first perfect $L_p$ sampler for $p>2$ for turnstile streams. 
That is, our algorithm samples a coordinate $i\in[n]$ with probability $\frac{|x_i|^p}{\|x\|_p^p}\pm\frac{1}{\poly(n)}$, where $x$ is defined through a turnstile stream. 
Formally, our guarantees are:
\begin{theorem}
\thmlab{thm:lp:perfect}
For any $p>2$ and failure probability $\delta\in(0,1)$, there exists a perfect $L_p$ sampler on a turnstile stream that uses $\tO{n^{1-2/p}\log\frac{1}{\delta}}$ bits of space and succeeds with probability at least $1-\delta$. Moreover, it obtains a $(1+\eps)$-estimation to the sampled item with probability at least $1-\delta$ using $\tO{\eps^{-2}n^{1-2/p}\log\frac{1}{\delta}}$ bits of space.
\end{theorem}
By comparison, the existing perfect $L_p$ sampler of \cite{JayaramW18} handles $p\le 2$. 
Their techniques rely on duplicating each coordinate $i\in[n]$ a total of $N=n^c$ times for some sufficiently large constant $c>1$ and then performing a separate scaling of each of the $n^{c+1}$ coordinates. 
Ultimately these coordinates are then hashed into a CountSketch data structure~\cite{CharikarCF04}, which has an error roughly on the order of the $L_2$ norm of the input vector and uses space logarithmic in the universe size. 
Note that $\log N=\O{\log n}$ and thus the algorithm of \cite{JayaramW18} uses space that is polylogarithmic in $n$. 
Generalizing this to $p>2$ would require an error roughly on the order of the $L_p$ norm of the input vector, which would use space \emph{polynomial} in the universe size. 
Unfortunately, the universe size after duplication is $N\gg n$, and so the resulting data structure would use space significantly larger than $n$, which is pointless for our purposes because we could just maintain the entire vector using linear space. 

Another point of comparison is the existing truly perfect samplers. 
The truly perfect sampler of \cite{JayaramWZ22} uses space $\tO{n^{1-1/p}}$ for $p>1$, which is prohibitively large for our purposes. 
On the other hand, the truly perfect sampler of \cite{PettieW25} uses $\O{\log n}$ bits of space, but can only handle $p<1$ and requires the random oracle model. 
Moreover, both of these truly perfect samplers can only be implemented in the insertion-only model, as opposed to the more general turnstile setting of \thmref{thm:lp:perfect}. 
Therefore, we require new techniques in achieving \thmref{thm:lp:perfect}.

\paragraph{Approximate $L_p$ samplers.}
We also present the first approximate $L_p$ sampler for $p>2$ for turnstile streams, which samples an index $i \in[n]$ with probability $\frac{|x_i|^p}{\|x\|_p^p}\cdot(1\pm \eps)$,
and outputs FAIL with probability at most $0.1$. Formally, our guarantees are:
\begin{theorem}
\thmlab{thm:lp:approximate}
For any $p>2$ and accuracy parameter $\eps \in (0,1)$, there exists an approximate $L_p$ sampler on a turnstile stream that uses
$n^{1-2/p} \log^2 n \log \frac{1}{\eps}\cdot\poly(\log\log(n))$ bits of space to run and has update time $\frac{1}{\eps} \cdot \polylog\left(n,\frac{1}{\eps}\right)$. In addition, it gives a $(1+\eps)$-estimation to the sampled item using extra $\frac{1}{\eps^2} n^{1-2/p} \log^2 n \log \frac{1}{\eps}\cdot\poly(\log\log(n))$ bits of space.
\end{theorem}
We complete our discussion by providing a sketching lower bound for the approximate $L_p$ samplers, which shows that our algorithm has space optimality in both $n$ and $\log n$ factors.
\begin{theorem}
Let $x \in \mathbb{R}^n$ be a vector. Suppose that there is a linear sketch that outputs an index $i \in [n]$ with probability $\frac{|x_i|^p}{\|x\|_p^p}\cdot (1\pm 0.01)$, and outputs FAIL with probability at most $0.1$. Then, its sketching dimension is at least $\Omega\left(n^{1-2/p}\log n\right)$.
\end{theorem}
A comparison of related work for sampling on data streams and our proposed samplers is displayed in \tableref{table:summary}.
\begin{table}[!htb]
\centering
{
\tabulinesep=1.1mm
\resizebox{\columnwidth}{!}{
\begin{tabu}{|c|c|c|c|c|}\hline
Sampler & Data Stream & Distortion & Randomness & Remarks\\\hline
\cite{MonemizadehW10} & Turnstile & Approximate & Standard & $L_p$, $p\in[0,2]$\\\hline
\cite{AndoniKO11} & Turnstile & Approximate & Standard & $L_p$, $p\in[0,2]$\\\hline
\cite{JowhariST11} & Turnstile & Approximate & Standard & $L_p$, $p\in[0,2]$\\\hline
\cite{JayaramW18} & Turnstile & Perfect & Standard & $L_p$, $p\in[0,2]$\\\hline
\cite{CohenG19} & Insertion-Only & Approximate & Standard & Soft Concave Sublinear \\\hline
\cite{JayaramWZ22} & Insertion-Only & Truly Perfect & Standard & $L_p$, $p\ge 1$; $M$-estimators \\\hline
\cite{PettieW25} & Insertion-Only & Truly Perfect & Random-Oracle Model & $L_p$, $p\in(0,1)$; L\'{e}vy processes \\\hline\hline
Our Work & Turnstile & Perfect & Standard & $L_p$, $p>2$; polynomials \\\hline
Our Work & Turnstile & Approximate & Standard & $L_p$, $p>2$\\\hline
\end{tabu}
}
}
\caption{Summary of related work for sampling on data streams}
\tablelab{table:summary}
\end{table}

\paragraph{General samplers.}
Another consequence of the existing techniques is that known perfect samplers for turnstile streams only handle the class of functions $G(z)=|z|^p$, for which the probability of sampling a coordinate $i$ is $\frac{x_i^p}{\|x\|_p^p}=\frac{(\alpha x_i)^p}{\|\alpha x\|_p^p}$, for any parameter $\alpha>0$. 
Thus the probability of sampling $i\in[n]$ from the vector $x$ is the same as the probability sampling $i\in[n]$ from the vector $\alpha x$ for any scalar $\alpha>0$. 
Unfortunately, many interesting functions $G(z)$, such as general polynomials, are not scale-invariant. 
Using our new techniques, we achieve perfect samplers for a wider class of functions:
\begin{theorem}
\thmlab{thm:lp:perfect:gen}
For any polynomial $G(z) = \sum_{d \in [D]} \alpha_d z^{p_d}$ with $0 <  p_1 < p_2 < \ldots < p_D = p$ and $0<\alpha_d<M$ for all $d \in [D]$, where $M$ and $D$ are some fixed constants, there exists a perfect polynomial sampler on a general turnstile stream, which outputs $i^* \in [n]$ with $\PPr{i^*=i}=\frac{G(x_i)}{\sum_{j \in [n]} G(x_j)} +\frac{1}{\poly(n)}$. 
The algorithm uses $\tO{n^{\max \{0, 1-2/p\}}\cdot\log \frac{1}{\delta}}$ bits of space and succeeds with probability at least $1-\delta$.
\end{theorem}
Indeed, \thmref{thm:lp:perfect:gen} includes functions such as polynomials that are not scale-invariant. 
Our techniques can also be used to obtain perfect $G$-samplers for the cap function $G(z)=\min(T,z^p)$ and the logarithmic function $G(z)=\log(1+z)$. 
Previously, approximate and perfect samplers for these functions were considered by \cite{CohenG19,PettieW25} for the insertion-only setting, but there were no known perfect samplers for the turnstile model. 
See \tableref{table:summary} for a summary of these contexts. 

\paragraph{Application to norm estimation.}
Next, we describe an interesting application of our perfect $L_p$ samplers to the problem of norm estimation, where the goal is to estimate $\|x\|_p=\left(x_1^p+\ldots+x_n^p\right)^{1/p}$ up to a $(1+\eps)$-multiplicative factor, for some input accuracy parameter $\eps\in(0,1)$. 
Note that up to constants, the problem is equivalent to estimating the $p$-th frequency moment of the vector, defined by $F_p(x)=x_1^p+\ldots+x_n^p$ up to a $(1+\O{\eps})$-multiplicative approximation. 
The estimation of norms/frequency moments is a fundamental problem for the streaming model. 
Indeed, since the seminal paper of Alon, Matias, and Szegedy~\cite{AlonMS99}, there has been a long line of research analyzing the space or time complexity of this problem~\cite{ChakrabartiKS03,Bar-YossefJKS04,Woodruff04,IndykW05,Indyk06,Li08,KaneNW10,KaneNPW11,Ganguly11,BravermanO13,BravermanKSV14,Andoni17,BlasiokDN17,BravermanVWY18,GangulyW18,WoodruffZ21b,WoodruffZ21,JayaramWZ24,BravermanZ24}. 

More generally, it is often desirable to understand the behavior of certain subsets of coordinates, e.g., iceberg queries in databases, range queries in computational geometry, etc. 
However, the identity of these coordinates may not be known as an input parameter.  
Formally, the goal is to estimate $\sum_{i\in\calQ} x_i^p$, for a subset $\calQ\subseteq[n]$ that is queried only on the data structure after the stream is processed. 
In this setting, many of the existing approaches are either suboptimal or fail altogether. 
For example, it is not clear how to adapt algorithms based on linear sketches to only consider the coordinates in $\calQ$. 
Similarly, approaches based on subsampling and heavy-hitters, e.g.,~\cite{IndykW05} are suboptimal. 
Our perfect $L_p$ samplers can achieve the following near-optimal guarantees:
\begin{restatable}{theorem}{thmlpforget}
\thmlab{thm:lp:forget}
Given $p>2$, there exists an algorithm that processes a turnstile stream defining a vector $x\in\mathbb{R}^n$ and a post-processing query set $\calQ\subseteq[n]$, and with probability at least $0.99$, outputs a $(1+\eps)$-approximation to $\|x_{\calQ}\|_p^p$. 
For $\|x_{\calQ}\|_p^p\ge\alpha\|x\|_p^p$, the algorithm uses $\tO{\frac{1}{\alpha\eps^2}n^{1-2/p}}$ bits of space. 
\end{restatable}
Naively, CountSketch requires roughly $\frac{1}{\alpha^2 \varepsilon^2} \cdot n^{1-2/p}$ space to ensure the estimation error is below $\varepsilon \cdot \|x_Q\|_p^p$, to achieve a $(1+\varepsilon)$-approximation to $\|x_Q\|_p^p$. By comparison, our algorithm uses $\frac{1}{\alpha \varepsilon^2} \cdot n^{1-2/p}$ space, which is better by a factor of $\frac{1}{\alpha}$. A more sophisticated approach is given by Proposition 5.1 in \cite{MintonP14}. However, this approach can only handle queries on subsets $Q$ with small size, e.g., polylogarithmic in $n$. 

It is known that even for the setting where $\calQ=[n]$, any $(1+\eps)$-approximation algorithm for $F_p$ estimation on insertion-only streams requires $\Omega\left(\frac{1}{\eps^2}n^{1-2/p}\right)$ bits of space~\cite{WoodruffZ21b}. 
Thus, our algorithm in \thmref{thm:lp:forget} is optimal up to polylogarithmic factors in $n$ and $\frac{1}{\eps}$. 
We also provide further optimizations to achieve fast update time. 

A specific application of our setting is the ``right to be forgotten data streaming'' (RFDS) model, recently introduced by \cite{Pavan0VM24}. 
Motivated by the right of any entity to decide whether their personal data should remain within a specific dataset, the RFDS model permits a forget operation in the data stream, which sets $x_i=0$ for an input $i\in[n]$. 
Although \cite{Pavan0VM24} shows that the RFDS model is in general difficult, \cite{lnsw25} observed that $L_p$-sampling is useful in the RFDS model, i.e., the streaming model with forget requests, where it put forth the idea of taking $\O{\frac{1}{\alpha\eps^2}}$ perfect $L_p$-samples (in their notation $\alpha$ is replaced with $1-\alpha$) and obtaining unbiased estimates for them using a Taylor series and averaging them. 
In fact, \cite{lnsw25} solves the harder problem, where forget requests can occur multiples times for an item throughout the course of a stream. 
Here we only allow the coordinates of entities who wish to have their information expunged to be requested after all the data is curated, i.e., at the end of the stream. 
However, one advantage of only allowing forget requests to occur at the end of the stream is that this allows us to solve the problem in the turnstile streaming model, whereas \cite{lnsw25} shows that if forget requests can occur before the end of the stream then no sublinear space algorithms are possible in the turnstile streaming model.  

\subsection{Motivation and Applications}
\paragraph{Statistical indistinguishability.}
Sampling is a fundamental primitive for extracting key information from large datasets. 
In particular, $L_p$ sampling has been used as a subroutine toward central data stream problems such as heavy-hitters, norm/moment estimation, cascaded norm estimation, duplicate detection and identification, and data summarization~\cite{MonemizadehW10,AndoniKO11,JowhariST11,BravermanOZ12,JayaramW18}. 
For example, $L_1$ sampling is used to extract a number of samples, thus generating a histogram that subsequently serves as a representative summary of the dataset, which is then the input to more complex downstream algorithms~\cite{GibbonsM98,GilbertKMS01,CormodeMR05,HuangNGHJJT07}, such as anomaly/event detection in network monitoring. 
Since these histograms effectively represent the entire dataset, it is important that these samplers capture the true distribution of the dataset with minimal bias or distortion. 
Unfortunately, approximate samplers have a relative error in their probabilities, consequently introducing potential statistical bias that propagates through the downstream algorithms. 
For example, algorithms that assume uniform sampling may experience inaccuracies due to this bias. 
These biases can be leveraged by a malicious attacker who adaptively queries a database for samples, which is the basis for the field of adversarial robust streaming~\cite{AvdiukhinMYZ19,ben2020adversarial,ben2020framework,hassidim2020adversarially,WoodruffZ21,BravermanHMSSZ21,AttiasCSS21,AjtaiBJSSWZ22,Ben-EliezerEO22,ChakrabartiGS22,AssadiCGS23,CherapanamjeriS23,WoodruffZZ23a,GribelyukLWYZ24,WoodruffZ24,GribelyukLWYZ25}. 
However, perfect $L_p$ samplers mitigate these issues by ensuring near-uniformity in their output distribution without increasing space complexity. 
In fact, even if we wish to extract $n^c$ samples from the true distribution for any constant $c>0$, we can set the additive distortion in the sampling probabilities to be $\frac{1}{n^{c+100}}$, so that the resulting total variation distance over the joint distribution of the samples remain statistically indistinguishable from extracting $n^c$ truly uniform samples, making them ideal for black-box applications. 

\paragraph{Distributed databases.}
In particular, an important application of sampling is in distributed databases, where independent samplers are locally implemented on disjoint portions of the dataset across multiple machines, which subsequently serve as compact summaries, providing valuable statistical insights into the distribution of data across the entire system. 
However, the accumulation of small biases from approximate samplers, manifested as variation distance from the true distribution, can present significant challenges, such as compromising the accuracy of downstream algorithms or sensitive statistical tests that rely on the fidelity of the sampled data. 
Perfect $L_p$ samplers can address these challenges by ensuring minimal distortion in their output, maintaining the integrity of both local and aggregate statistical summaries.

\paragraph{Privacy considerations.}
Another compelling motivation for the use of perfect $L_p$ samplers is their role in privacy-preserving applications. 
In such scenarios, the dataset $x\in \mathbb{R}^n$ is sampled to reveal an index $i\in [n]$ to an external party, while minimizing the leakage of global information about $x$. 
Since approximate samplers introduce a multiplicative bias into the sampling probabilities, which may depend on the global properties of the dataset, potential adversaries could exploit this bias to infer sensitive information about $x$.
For example, under such guarantees, it is permissible for a sampler to bias the sampling probabilities for a large set $S$ of indices by $(1+\eps)$ if a certain global property $\calP$ holds for $x$ and might instead bias the sampling probabilities of $S$ by $(1-\eps)$ if $\calP$ does not hold. 
Then an observer can deduce from a small number of samples whether the property $\calP$ holds by estimating the total sampling probabilities of indices in $S$. 
On the other hand, perfect $L_p$ samplers produce samples with polynomially small additive bias, reducing the potential for such leakage. 
This characteristic makes them better suited for privacy-sensitive tasks, where the goal is to reveal minimal information about the underlying data. 

\paragraph{Heavy-tailed emphasis.}
A key advantage of $L_p$ sampler for $p>2$ is its focus on dominant contributions.
When analyzing the vector $p$-norm $\|x\|_p$ where $p > 2$, heavier emphasis is placed on the coordinates of the vector $x$ that have larger frequencies. 
This property makes the $p$-norm particularly useful in scenarios where the focus is on prioritizing elements with larger contributions, such as in sparse signal recovery, outlier detection, and high-dimensional data analysis. 
More generally, the $p$ parameter can be interpreted as a interpolation between $L_0$, where all coordinates have the same contribution, regardless of their magnitude, and $L_\infty$, where only the largest coordinate is relevant. 

\subsection{Preliminaries}
In this paper, we use the notation $[n]$ to denote the set $\{1,2,\ldots,n\}$ for an integer $n\ge 1$. 
We use the notation $\poly(n)$ to denote a fixed polynomial whose degree can be determined by setting the appropriate constants in the algorithm and analysis. 
We similarly use the notation $\polylog(n)$ to denote $\poly(\log n)$. 
We say an event occurs with high probability if it occurs with probability at least $1-\frac{1}{\poly(n)}$. 
For a possibly multivariate function $f$, we use the notation $\tO{f}=f\cdot\polylog(f)$. 
For a vector $x\in\mathbb{R}^n$, we define the $p$-norm of $x$ by $\|x\|_p=\left(x_1^p+\ldots+x_n^p\right)^{1/p}$ and we define the $p$-th moment of $x$ by $F_p(x)=\|x\|_p^p$. 

We recall the following formulation of the Khintchine inequality:
\begin{theorem}[Khintchine inequality]
\thmlab{thm:khintchine}
\cite{haagerup1981best}
Let $r_1,\ldots,r_n\in\{-1,+1\}$ be independent random signs and let $p\ge 2$. 
Then  
\[\Ex{|r_1x_1+\ldots+r_nx_n|^p}\le(B_p)^p\cdot\|x\|_2^p,\]
where $B_p=\sqrt{2}\cdot\left(\frac{1}{\sqrt{\pi}}\cdot\Gamma\left(\frac{p+1}{2}\right)\right)^{1/p}$ and $\Gamma$ is the Gamma function.
\end{theorem}
Moreover, we recall the following property of the Gamma function.
\begin{proposition}
\proplab{prop:gamma}
$\left(\Gamma\left(\frac{p+1}{2}\right)\right)^{1/p}=\Theta(\sqrt{p})$.
\end{proposition}
We formally define the notion of a perfect $L_p$ sampler as follows:
\begin{definition}[Perfect $L_p$ sampler]
Given a turnstile data stream $S$, let $x \in \mathbb{R}^{n}$ be the vector induced by $S$. 
A perfect $L_p$ sampler outputs an index $i^*$, such that for all $i \in [n]$, we have
\[\PPr{i^*=i} = \frac{|x_i|^p}{\|x\|_p^p} \pm \frac{1}{\poly (n)}.\]
The sampler is allowed to output FAIL with probability $\delta>0$, which is given by an input parameter and we generally set to $\frac{1}{3}$. 
\end{definition}
\cite{JayaramW18} introduced a perfect $L_p$ sampler for $p\in(0,2]$ with the following guarantees:
\begin{theorem}[Perfect $L_p$ sampler for $p \le 2$, c.f. Theorem 9 in \cite{JayaramW18}] 
\thmlab{thm:lp_sampler_p_small}
Given a turnstile data stream $S$, let $x \in \mathbb{R}^{n}$ be the vector induced by $S$. For $p \in (0,2]$, there exists a perfect $L_p$ sampler with failure probability at most $\delta_1$. Moreover, it outputs an estimate $\widehat{x}$ such that $\widehat{x}=(1 \pm \eps) x_i$ with probability $1-\delta_2$. The sampler uses
\[\O{\left(\log ^2n(\log \log n)^2+\beta \log n \log \left(1 / \delta_2\right)\right) \log \left(1 / \delta_1\right)}\]
bits of space for $p<2$, where $\beta=\min \left\{\eps^{-2}, \eps^{-p} \log \left(1/\delta_2\right)\right\}$, and
\[\O{\left(\log ^3n+\eps^{-2} \log ^2n \log \left(1 / \delta_2\right)\right) \log \left(1 / \delta_1\right)}\]
bits of space for $p=2$.
\end{theorem}

\paragraph{Exponential random variables.}
Throughout our work, we shall frequently use exponential random variables and a number of their properties. 
We first define an exponential random variable:
\begin{definition}[Exponential random variable]
If $\be$ is a exponential random variable with scale parameter $\lambda>0$, then the probability density function for $\be$ is 
\[p(x)=\lambda e^{-\lambda x}.\]
We say $\be$ is a standard exponential random variable if $\lambda=1$.
\end{definition}
We have the following facts about exponential random variables. 
\begin{proposition}
Let $\be$ be a standard exponential random variable. 
Then for any $a,b\ge 0$, 
\begin{align*}
\Pr[ \be \ge a\log n ] = \frac{1}{n^a},~~~ \Pr[ \be\le b ] \le b.
\end{align*}
\end{proposition}

\begin{proposition}[Scaling of exponentials]
\proplab{prop:exp_scaling}
    Let $t$ be exponentially distributed with rate $\lambda$, and let $\alpha>0$. Then $\alpha t$ is exponentially distributed with rate $\lambda / \alpha$.
\end{proposition}
We define an anti-rank vector to be the ranking of the indices based on their magnitudes, generally in the context of after scaling by exponential random variables. 
Formally, for a vector $z\in\mathbb{R}^n$ and for $k\in[n]$, we define the $k$-th anti-rank $D(k) \in[n]$ of $z$ to be the index $D(k)$ so that $|z_{D(1)}|\ge\ldots\ge|z_{D(n)}|$. 
Using the structure of the anti-rank vector of a set of exponential random variables, \cite{Nag06} introduces a simple form for describing the distribution of $t_{D(k)}$ as a function of $\left(\lambda_1, \ldots, \lambda_n\right)$ and the anti-rank vector.

\begin{proposition}[\cite{Nag06}]
\proplab{prop:exp_max_prob}
 For any $i=1,2, \ldots, n$, we have
\[\operatorname{Pr}[D(1)=i]=\frac{\lambda_i}{\sum_{j=1}^n \lambda_j}\]
\end{proposition} 

\begin{proposition}[\cite{Nag06}]
\proplab{prop:exp_order_statistics}
Let $\left(t_1, \ldots, t_n\right)$ be independently distributed exponentials, where $t_i$ has rate $\lambda_i>0$. Then for any $k=1,2, \ldots, n$, we have
\[t_{D(k)}=\sum_{i=1}^k \frac{E_i}{\sum_{j=i}^n \lambda_{D(j)}},\]
where $E_1, E_2, \ldots, E_n$ are i.i.d. exponential random variables with mean $1$ that are independent of the anti-rank vector $(D(1), D(2), \ldots, D(n))$.
\end{proposition} 
\cite{JayaramW18} showed the following characterization of each coordinate $z_{D(k)}$ under the scaling $z_i=\frac{x_i}{\be_i^{1/p}}$, where $\be_i$ is an independent exponential random variable for each $i\in[n]$. 
\begin{lemma}[\cite{JayaramW18}]
\lemlab{lem:exp_inverse}
Let $f \in \mathbb{R}^n$ be a vector, let $\left(\be_1, \ldots, \be_n\right)$ be i.i.d. exponential random variables with rate $1$, let $z_i=x_i / \be_i^{1 / p}$, and let $(D(1), \ldots, D(n))$ be the anti-rank vector of the vector $\left(|z_1|^{-p}, \cdots |z_n|^{-p}\right)$. 
Then we have
\[\PPr{D(1)=i} = \frac{\left|x_i\right|^p}{\|x\|_p^p}.\]
As a result, the probability that $\left|z_i\right|=\arg \max _j\left\{\left|z_j\right|\right\}$ is precisely $\left|x_i\right|^p /\|x\|_p^p$, so for a perfect $L_p$ sampler it suffices to return $i \in[n]$ with $\left|z_i\right|$ maximum.
Moreover, we have
\[z_{D(k)}=\left(\sum_{i=1}^k \frac{E_i}{\sum_{j=i}^N f_{D(j)}^p}\right)^{-1 / p},\]
where $E_i$ 's are i.i.d. exponential random variables with mean $1$, and are independent of the anti-rank vector $(D(1), \ldots, D(n))$.
\end{lemma}

The next statement shows that the maximum scaled vector is roughly a $\frac{1}{\log^2 n}$-heavy hitter with respect to the entire scaled vector. 
\begin{lemma}
[\cite{EsfandiariKMWZ24}]
\lemlab{lem:max:heavy}
Let $\be_1,\ldots,\be_n$ be independent standard exponential random variables, let $\alpha_1,\ldots,\alpha_n \ge 0$, and let $C>0$ be a fixed constant.  
Then
\begin{align*}
\Pr \Big[ \frac{\max_{i\in[n]}\alpha_i/\be_i}{\sum_{i=1}^n \alpha_i/\be_i}\ge\frac{1}{C\log^2 n} \Big] \ge1-\frac{1}{\poly (n)}.
\end{align*}
\end{lemma}
Next, we show that the coordinate with the largest magnitude can be related to the $L_p$-norm of the original input vector. 
\begin{lemma}
\lemlab{lem:max_stability}
Let $f \in \mathbb{R}^n$ be a vector, let $\left(\be_1, \ldots, \be_n\right)$ be i.i.d. exponential random variables with rate $1$, let $z_i=x_i / \be_i^{1 / p}$, and let $(D(1), \ldots, D(n))$ be the anti-rank vector of the vector $\left(|z_1|^{-p}, \cdots |z_n|^{-p}\right)$. Then, we have that $|z_{D(1)}| > \frac{\|x\|_p}{100 \log \frac{1}{\eps}}$ holds with probability $1 - \poly(\eps)$.
\end{lemma}
\begin{proof}
By \lemref{lem:exp_inverse}, we have
\[|z_{D(1)}|= \frac{\|x\|_p}{\be},\]
where $\be$ is a random exponential variable with rate $1$, independent of $(D(1), \cdots, D(n))$. Notice that by the cdf of exponential variables, we have
\[\PPr{\be < \log 100 \cdot \frac{1}{\eps}} > 1- \poly(\eps),\]
which proves our desired result.
\end{proof}
We upper bound the $L_2$ norm of the scaled vector $z$ in terms of the original input vector $f$. 
\begin{lemma}
\lemlab{lem:z_second_moment}
Let $f \in \mathbb{R}^n$ be a vector, let $\left(\be_1, \ldots, \be_n\right)$ be i.i.d. exponential random variables with rate $1$, and let $z_i=x_i / \be_i^{1 / p}$. Then, we have $\|z\|_2 = \O{\|f\|_2}$ with probability $\Omega(1)$.
\end{lemma}
\begin{proof}
Now we have
\[\|z\|_2^2 = \sum_{i \in [n]} \frac{x_i^2}{\be_i^{2/p}}.\]
where the $\be_i$'s are i.i.d. exponential random variables with rate $1$. Taking the expectation of $\|z\|_2^2$, we have
\begin{align*}
\Ex{\|z\|_2^2} = \sum_{i \in [n]} \Ex{\frac{x_i^2}{\be_i^{2/p}}} = \sum_{i \in [n]} x_i^2 \cdot \Ex{\frac{1}{\be^{2/p}}}.
\end{align*}
Consider $\Ex{\frac{1}{\be^{2/p}}}$, let $g(x)$ denote the pdf of $\be$, we have
\begin{align*}
\Ex{\frac{1}{\be^{2/p}}} &= \int_0^\infty t^{-2/p} g(t) \d t \\
&= \int_0^1 t^{-2/p} e^{-t} \d t +\int_1^\infty t^{-2/p} e^{-t} \d t .
\end{align*}
For the first term, since $e^{-t} \le 1$ we have
\[\int_0^1 t^{-2/p} e^{-t} \d t \le \int_0^1 t^{-2/p} \d t = \frac{t^{1-2/p}}{1 - 2/p}  ~\Big|_0^1 = \O{1}.\]
For the second term, since $p > 2$, we have $t^{-2/p} \le 1$ for $t \ge 1$. Thus, we have
\[\int_1^\infty t^{-2/p} e^{-t} \d t \le \int_1^\infty e^{-t} \d t = - e^{-t} ~\Big|_1^\infty = \O{1}.\]
Combining the bounds, we have $\Ex{\|z\|_2^2} = \O{\|f\|_2^2}$. Then, by Markov's inequality, we have the desired guarantee.
\end{proof}

\section{Perfect Sampling}
\seclab{sec:perfect:sampling}
In this section, we describe our implementations for both the perfect $L_p$ samplers and the perfect polynomial samplers. 



\subsection{Integer \texorpdfstring{$p$}{p}}
We first provide the intuition for our perfect $L_p$ sampler for integer $p>2$. 
A natural starting point would be to adapt the techniques for existing perfect $L_p$ samplers with $p\in(0,2]$. 
The only existing implementation requires duplicating each coordinate a polynomial number of times to utilize the max-stability property of exponential random variables. 
To identify the maximum coordinate, \cite{JayaramW18} only needs polylogarithmic space to find the $L_2$-heavy hitters, but to find the $L_p$-heavy hitters for $p>2$, the space required is polynomial in the universe size, which is now substantially large due to the duplication. 

Instead, we use perfect $L_2$ samplers as a black-box subroutine to extract a coordinate $i\in[n]$. 
We would like to output $i$ with probability $\frac{x_i^p}{\|x\|_p^p}+\frac{1}{\poly(n)}$ and we have sampled $i$ with probability roughly $\frac{x_i^2}{\|x\|_2^2}+\frac{1}{\poly(n)}$. 
Thus, conditioned on the $L_2$ sampler outputting $i$, we would like to output $i$ with probability $x_i^{p-2}\cdot\frac{F_2(x)}{F_p(x)}$.  
Unfortunately, this expression may not be a well-defined probability because it may be larger than $1$, for instance if $x_i=n^{1/p}$ and $F_2(x)=F_p(x)=\Theta(n)$. 
Therefore, we instead would like to only output $i$ with probability $x_i^{p-2}\cdot\frac{F_2(x)}{n^{1-2/p}F_p(x)}$. 
Although we do not have each of the terms $x_i^{p-2}$, $F_2(x)$, and $F_p(x)$, we can obtain constant-factor approximations to $F_2(x)$ and $F_p(x)$ using existing procedures~\cite{AlonMS99,Ganguly15}. 
It thus remains to estimate $x_i^{p-2}$. 

In fact, the index returned by the perfect $L_2$ sampler of \cite{JayaramW18} is the largest scaled coordinate $\frac{x_i}{\sqrt{\be_i}}$, where $\be_i$ is an independent exponential random variable for all $i\in[n]$ and in particular, a heavy-hitter of the resulting scaled vector. 
We can thus acquire an unbiased estimate $\widehat{x_i^{p-2}}$ to $x_i^{p-2}$ by running $(p-2)$ independent instances of CountSketch on the scaled coordinates, which has a small relative variance since $\frac{x_i}{\sqrt{\be_i}}$ is a heavy-hitter. 
However, this is still not enough because there is a non-trivial probability that $\widehat{x_i^{p-2}}\cdot\frac{F_2(x)}{F_p(x)}$ is still larger than $1$ if the estimate is procured through $(p-2)$ instances of CountSketch alone. 
Hence we further show that by the Khintchine inequality, a sufficiently tight approximation of $x_i^{p-2}$ can be obtained using $\polylog(n)$ instances of CountSketch. 
Finally, we show that with $n^{1-2/p}\cdot\polylog(n)$ number of perfect $L_2$ samples, one of these samples will be passed through the subsequent rejection sampling. 
Our algorithm appears in full in \algref{alg:lp:sampler}. 

\begin{algorithm}[!htb]
\caption{Perfect $L_p$ sampler for integer $p>2$}
\alglab{alg:lp:sampler}
\begin{algorithmic}[1]
\Require{Input vector $x\in\mathbb{R}^n$ in a stream}
\Ensure{Perfect $L_p$ sample}
\State{Let $C$ be the constant from \corref{cor:sample:tail} and let $N= \O{n^{1-2/p}}$}
\State{$s_1,\ldots,s_N$ be $N$ perfect $L_2$ samples from $x$} \Comment{See \thmref{thm:lp_sampler_p_small}}
\State{Use $\AMS$ to get a $2$-approximation $\widehat{F_2}$ to $F_2(x)$}
\State{Use $\FpEst$ to get a $2$-approximation $\widehat{F_p}$ to $F_p(x)$}
\For{each $i\in[N]$}
\State{Let $j$ be the index of $s_i$}
\For{$a \in [p-2]$}
\State{Run $\polylog(n)$ instances of $\CountSketch$ to acquire estimates $\widehat{x^{(q,1)}_j},\ldots,\widehat{x^{(q,\polylog(n))}_j}$}
\State{$\widehat{x_j^{(q)}} \gets \frac{1}{\polylog(n)} \sum_{l \in [\polylog(n)]} \widehat{x^{(q,l)}_j}$}
\EndFor
\State{\Return $j$ with probability $\frac{\widehat{F_2}}{8 n^{1-2/p}\cdot\widehat{F_p}}\cdot\prod_{a\in[p-2]} \left|\widehat{x^{(a)}_j}\right|$}
\State{Otherwise, continue to next $i$}
\EndFor
\end{algorithmic}
\end{algorithm}

To analyze our algorithm, we first show that the magnitude of a signed sum of coordinates can be bounded in terms of the $L_2$ of the vector with high probability. 
\begin{lemma}
\lemlab{lem:cs_bounded_est_mean}
Let $m$ be some large constant to be determined later. For all $l \in [\polylog (n)]$, let $r_1^{(l)},\ldots,r_n^{(l)}\in\{-1,+1\}$ be independent random signs. 
Then we have,
\[\left|\frac{1}{\log^m (n)} \cdot \sum_{l \in [\log^m (n)]} r_1^{(l)}x_1+\ldots+r_n^{(l)}x_n\right|\le \frac{1}{\log^{m/4}(n)} \cdot \|x\|_2.\]
with probability at least $1-\frac{1}{\poly(n)}$. 
\end{lemma}
\begin{proof}
Let $m$ be a large enough constant, by the linearity of expectation, we have
\begin{align*}
\Ex{\left|\frac{1}{\log^m (n)} \cdot \sum_{l \in [\log^m (n)]} r_1^{(l)}x_1+\ldots+r_n^{(l)}x_n\right|^p} &= \frac{1}{\log^{mp}(n)} \cdot \Ex{\left| \sum_{l \in [\log^m (n)]} r_1^{(l)}x_1+\ldots+r_n^{(l)}x_n\right|^p}.
\end{align*}
Then, by \thmref{thm:khintchine} and \propref{prop:gamma}, we have that
\begin{align*}
\Ex{\left|\frac{1}{\log^m (n)} \cdot \sum_{l \in [\log^m (n)]} r_1^{(l)}x_1+\ldots+r_n^{(l)}x_n\right|^p}
&\le\frac{1}{\log^{mp}(n)} \cdot (B\cdot\sqrt{p})^p\cdot ( \log^m (n) \cdot \|x\|_2^2)^{p/2} \\
&\le\frac{1}{\log^{(m/2)\cdot p} n}(B\cdot\sqrt{p})^{p}\cdot\|x\|_2^p.
\end{align*}
for some absolute constant $B>0$. 
Then for any constant $c>0$, we have by Markov's inequality,
\[\PPr{\left|\frac{1}{\log^m (n)} \cdot \sum_{l \in [\log^m (n)]} r_1^{(l)}x_1+\ldots+r_n^{(l)}x_n\right|^p\ge \frac{n^c}{\log^{(m/2)\cdot p} n}\cdot(B\cdot\sqrt{p})^p\cdot\|x\|_2^p}\le\frac{1}{\poly(n)}.\]
Thus for $p=\log(n)$, we have
\[\PPr{\left|\frac{1}{\log^m (n)} \cdot \sum_{l \in [\log^m (n)]} r_1^{(l)}x_1+\ldots+r_n^{(l)}x_n\right|\ge \frac{2^c}{\log^{m/2} n}\cdot(B\cdot\sqrt{\log n})\cdot\|x\|_2}\le\frac{1}{\poly(n)}.\]
\end{proof}
As a result, it follows that $\CountSketch$ can be used to give a good additive estimate $\widehat{x_i}$ to each coordinate $x_i$. 
\begin{corollary}
\corlab{cor:sample:tail:init}
For each $i\in[n]$, the mean of $\log^m(n)$ instances of $\CountSketch$ gives an estimate $\widehat{x_i}$ such that with probability at least $1-\frac{1}{\poly(n)}$
\[\max_{i\in[n]}|\widehat{x_i}-x_i|\le \frac{1}{\log^{m/4}(n)}\cdot\|x\|_2.\]
\end{corollary}
Similarly, we can bound the error of the estimated value of the sampled coordinate. 
\begin{corollary}
\corlab{cor:sample:tail}
For the index $i$ output by a perfect $L_2$ sample, we have an estimate $\widehat{x_i}$ such that with probability at least $1-\frac{1}{\poly(n)}$, 
\[|\widehat{x_i}-x_i|\le \frac{1}{\polylog(n)}\cdot|x_i|.\]
\end{corollary}
\begin{proof}
First, we describe how we estimate the sampled entry. In the $L_p$ sampler introduced by \cite{JayaramW18}, for a vector $x$, we scale each entry by $g_i := \frac{x_i}{\be_i^{1/2}}$ where $\be_i$'s are i.i.d. exponential random variables with rate $1$, and we find the maximum entry of $g$. Notice that we have the following property of exponential variables (c.f. \lemref{lem:max:heavy}):
\[\PPr{\max g_i^2 \ge \frac{1}{C \log^2 n} \|g\|_2^2} \ge 1-\frac{1}{\poly(n)}.\]
Therefore, we implement $\CountSketch$ on the scaled vector $g$ to get an estimation of the sampled item $g_i$. By \corref{cor:sample:tail:init}, it has error at most $\frac{1}{\log^{m/4}(n)} \cdot \|g\|_2$ with probability at least $1-\frac{1}{\poly(n)}$. Hence, the error is bounded by $\frac{1}{\polylog(n)} \cdot g_i$ with probability at least $1-\frac{1}{\poly(n)}$ for sufficiently large $m$. We output $g_i\cdot \be_i$ as an estimation to $x_i$, so that the error is bounded by $\frac{1}{\polylog(n)} \cdot x_i$ with probability at least $1-\frac{1}{\poly(n)}$.
\end{proof}

We now show that our algorithm returns a random coordinate under the correct probability distribution for $L_p$ sampling. 
\begin{lemma}
\lemlab{lem:lp:int:correct}
With high probability, \algref{alg:lp:sampler} outputs an index $i\in[n]$. 
Moreover, for each $j\in[n]$, we have that
\[\PPr{i=j}=\frac{|x_j|^p}{\|x\|_p^p}\pm\frac{1}{\poly(n)}.\]
\end{lemma}
\begin{proof}
For each $i\in[N]$, we sample $j\in[n]$ with probability $\frac{|x_j|^2}{\|x\|_2^2}\pm\frac{1}{\poly(n)}$ by the correctness of perfect $F_2$ sampling. 
For each $a\in[p-2]$, $\widehat{x_j^{(a)}}$ is an unbiased estimate of $x_j$ such that with high probability, $|x_j-\widehat{x_j^{(a)}}|\le\frac{1}{\polylog(n)}\cdot|x_j|$. 
Then the index $j$ is returned with probability
\[\left(\frac{|x_j|^2}{\|x\|_2^2}\pm\frac{1}{\poly(n)}\right)\cdot\frac{\widehat{F_2}}{N\cdot\widehat{F_p}}\cdot\prod_{a\in[p-2]}\widehat{x^{(a)}_j}.\]
Let $\calE_1$ be the event that $|x_j-\widehat{x_j^{(a)}}|\le\frac{1}{\polylog(n)}\cdot|x_j|$ for all $a\in[p-2]$ so that by \corref{cor:sample:tail} and a union bound, $\PPr{\calE_1}\ge1-\frac{1}{\poly(n)}$. Conditioned on $\calE_1$, we have that
\[\prod_{a\in[p-2]}\widehat{x^{(a)}_j} = \prod_{a\in[p-2]}x_j \cdot \left(1 \pm \frac{1}{\polylog(n)}\right).\] Let $\calE_2$ be the event that $\AMS$ and $\FpEst$ return $2$-approximations of $\|x\|_2^2$ and $\|x\|_p^p$ respectively. 
Conditioned on $\calE_1$ and $\calE_2$, then we have that 
\begin{align*}
\frac{\widehat{F_2}}{8 n^{1-2/p}\cdot\widehat{F_p}}\cdot\prod_{a\in[p-2]}|\widehat{x^{(a)}_j}|&\le\frac{|x_j|^{p-2}\cdot\|x\|_2^2}{n^{1-2/p}\cdot\|x\|_p^p}\\
&\le\frac{|x_j|^{p-2}\cdot\|x\|_2^2}{n^{1-2/p}\cdot\left(|x_1|^p+\ldots+|x_n|^p\right)^{1-2/p}\cdot\|x\|_p^2}\\
&\le\frac{|x_j|^{p-2}\cdot\|x\|_2^2}{n^{1-2/p}\cdot|x_j|^{p-2}\cdot\|x\|_p^2}\le 1.
\end{align*}
Therefore, our rejection probability is smaller than $1$, so it is well-defined. Moreover, we have $\Ex{\widehat{x^{(a)}_j}}=x_j$ for all $a\in[p-2]$. Notice that $x^{(a)}_j$ has the same sign as $x_j$ conditioned on $\calE_1$, we have $\Ex{|\widehat{x^{(a)}_j}|}=|x_j|$.
Thus in expectation, the probability that $j$ is returned is
\[\frac{|x_j|^{p}}{8 n^{1-2/p}\cdot\|x\|_p^p}\pm\frac{1}{\poly(n)}.\]
Hence, conditioned on some index being output, the probability that $j$ is returned is
\begin{align*}
\left(\frac{|x_j|^{p}}{8 n^{1-2/p}\cdot\|x\|_p^p}\pm\frac{1}{\poly(n)}\right)&\cdot\left(\sum_{\ell\in[n]}\frac{|x_\ell|^{p}}{8 n^{1-2/p}\cdot\|x\|_p^p}\pm\frac{1}{\poly(n)}\right)^{-1}\\
&=\frac{|x_j|^{p}}{\|x\|_p^p}\pm\frac{1}{\poly(n)},
\end{align*}
which is the correct probability distribution. 
Finally, the probability that some index being output for each sample $i\in[N]$ is
\[\sum_{j\in[n]}\frac{|x_j|^{p}}{8 n^{1-2/p}\cdot\|x\|_p^p}\pm\frac{1}{\poly(n)}\ge\frac{1}{16 n^{1-2/p}}.\]
Thus by repeating $N=\O{ n^{1-2/p}}$ times, we have that a sample is returned with probability at least $1-\frac{1}{\poly(n)}$. 
\end{proof}

Next, we analyze the space complexity of our algorithm. 
\begin{lemma}
\lemlab{lem:lp:int:space}
\algref{alg:lp:sampler} uses $n^{1-2/p}\cdot \polylog(n)$ bits of space.
\end{lemma}
\begin{proof}
For each $L_2$ sampler, we need its failure probability $\delta_1 $ to be $\frac{1}{\poly(n)}$, so that by a union bound, all $L_2$ samplers succeed with probability $1 - \frac{1}{\poly(n)}$. Then, by \thmref{thm:lp_sampler_p_small}, we uses $N \cdot \O{\log^4 n} = \O{n^{1-2/p} \cdot \log^4 n}$ to acquire the $L_2$ samples. Moreover, we implement $\polylog(n)$ instances of $\CountSketch$ to estimate the value of each sample, which uses $N \cdot \polylog(n)$ bits of space. Notice that the space consumption of AMS, \textsc{FpEst} and the $p-2$ instances of \textsc{CountSketch} is dominated by the $L_2$ samplers. Thus, we have the stated space complexity.
\end{proof}
Putting together \lemref{lem:lp:int:correct} and \lemref{lem:lp:int:space}, we obtain the full guarantees for our perfect $L_p$ sampler for integer $p>2$. 
\begin{theorem}
\thmlab{thm:lp_sampler_int}
Given an integer $p>2$, there exists a perfect $L_p$ sampler that uses $\tO{n^{1-2/p}}$ bits of space. 
\end{theorem}

\subsection{Fractional \texorpdfstring{$p$}{p}}
Next, we generalize our results from integer $p>2$ to general $p>2$. 
The main challenge in the previous techniques is that to estimate $x_i^{p-2}$ for integer $p>2$, an intuitive approach would be to acquire $p-2$ independent estimates for $x_i$. 
To estimate $x_i^{p-2}$ for fractional $p>2$, we utilize the Taylor series expansion of $x_i^{p-2}$. 
Our full algorithm appears in \algref{alg:lp:sampler_frac}. 
\begin{algorithm}[!htb]
\caption{Perfect $L_p$ sampler for general $p>2$}
\alglab{alg:lp:sampler_frac}
\begin{algorithmic}[1]
\Require{Input vector $x\in\mathbb{R}^n$ in a stream}
\Ensure{Perfect $L_p$ sample}
\State{Let $Q=\O{\log n}$ with large enough constant term}
\State{Let $N'= n^{1-2/p} \cdot \polylog(n)$. Let $N = N' \cdot \O{\log n}$}
\State{$s_1,\ldots,s_N$ be $N$ perfect $L_2$ samples from $x$} \Comment{See \thmref{thm:lp_sampler_p_small}}
\For{each $i \in [N]$}
\State{Let $y_{s_i}$ be the constant approximation to $x_{s_i}$ satisfying $y_{s_i} \in [\frac{99x_{s_i}}{100 e^p \log n},\frac{101x_{s_i}}{100 e^p \log n}]$}
\EndFor
\State{Use $\AMS$ to get a $2$-approximation $\widehat{F_2}$ to $F_2(x)$}
\State{Use $\FpEst$ to get a $2$-approximation $\widehat{F_p}$ to $F_p(x)$}
\For{each $i\in[N]$}
\State{Let $j$ be the index of $s_i$}
\For{$q \in [Q]$}
\State{Run $\polylog(n)$ instances of $\CountSketch$ to acquire estimates $\widehat{x^{(q,1)}_j},\ldots,\widehat{x^{(q,\polylog(n))}_j}$}
\State{$\widehat{x_j^{(q)}} \gets \frac{1}{\polylog(n)} \sum_{l \in [\polylog(n)]} \widehat{x^{(q,l)}_j}$}
\EndFor
\State{$\widehat{x_j^{p-2}} \gets \sum_{q=0}^Q \left(\binom{p-2}{q}y_j^{p-2-q} \cdot \prod_{a \in [q]} (\widehat{x_j^{(a)}}-y_j) \right)$} \Comment{See \lemref{lem:truncated_taylor_sampler}}
\State{\Return $j$ with probability $\frac{\widehat{F_2}}{4N'\cdot\widehat{F_p}}\cdot \left|\widehat{x_j^{p-2}}\right|$}
\State{Otherwise, continue to next $i$}
\EndFor
\end{algorithmic}
\end{algorithm}

In the analysis, we first show that a Taylor series expansion truncated at $Q=\O{\log n}$ terms is a good approximation to $x^p$ for any fixed constant $p>2$. 
\begin{lemma} \lemlab{lem:truncated_taylor_sampler}
Let $Q=\O{\log n}$ and $y_i \in\left[\frac{99x_i}{100},\frac{101x_i}{100}\right]$. 
Then 
\[\left\lvert x_i^p-\sum_{q=0}^Q \binom{p}{q}y_i^{p-q}(x_i-y_i)^q\right\rvert\le\frac{1}{\poly(n)}\cdot x_i^p.\]
\end{lemma}
\begin{proof}
Using the Taylor expansion we have
\[ x_i^p = \sum_{q=0}^\infty \binom{p}{q}y_i^{p-q}(x_i-y_i)^q.\]
Note that $|x_i - y_i| \le \frac{x_i}{100}$, so for $q > p$ we have
\[\left| \binom{p}{q}y_i^{p-q}(x_i-y_i)^q \right| 
\le \O{(2e)^{p/2} \cdot x_i^{p-q} \cdot \left(\frac{x_i}{100}\right)^q}=\O{x_i^p \cdot \left(\frac{1}{100}\right)^q},\]
where the first step is by $\binom{p}{q} \le (2e)^{p/2}$.
Hence, we have that for $Q = \O{\log n}$,
\[\sum_{q=Q+1}^\infty \binom{p}{q}y_i^{p-q}(x_i-y_i)^q \le \O{1} \cdot x_i^p \cdot \sum_{q=Q+1}^\infty \left(\frac{1}{100}\right)^q \le \frac{1}{\poly(n)} \cdot x_i^p.\]
\end{proof}
We show that our algorithm produces a sample $i$ according to the correct probability distribution. 
\begin{lemma}
\lemlab{lem:frac:correct}
With high probability, \algref{alg:lp:sampler_frac} outputs an index $i\in[n]$. 
Moreover, for each $j\in[n]$, we have that
\[\PPr{i=j}=\frac{|x_j|^p}{\|x\|_p^p}\pm\frac{1}{\poly(n)}.\]
\end{lemma}
\begin{proof}
By a similar argument as \thmref{thm:lp_sampler_int}, the index $j$ is returned with probability
\[\left(\frac{|x_j|^2}{\|x\|_2^2}\pm\frac{1}{\poly(n)}\right)\cdot\frac{\widehat{F_2}}{4N'\cdot\widehat{F_p}}\cdot\widehat{x_j^{p-2}}.\]
Let $\calE_1$ be the event that $|x_j-\widehat{x_j^{(a)}}|\le\frac{1}{\polylog(n)}\cdot|x_j|$ for all $a\in[N]$ so that by \corref{cor:sample:tail} and a union bound, $\PPr{\calE_1}\ge1-\frac{1}{\poly(n)}$. Conditioned on $\calE_1$, we have that for all $q \in [Q]$
\[\left|\left(\widehat{x^{(q)}_j} - y_j\right) - (x_j - y_j)\right|\le\frac{1}{\polylog(n)}\cdot|x_j|.\]
For $Q = \O{\log n}$ and $|x_j - y_j| \le \frac{x_j}{100p \log(n)}$, we have
\[\prod_{a \in [q]} (\widehat{x_i^{(a)}}-y_j) = \left((x_j-y_j)\pm \frac{x_j}{\polylog(n)}\right)^q = (x_j-y_j)^q  \pm \frac{y_j^q}{100p^q\log^q(n)}.\]
Recall our truncated Taylor estimator for $x^{p-2}$ is
\begin{align*}
\widehat{x_j^{p-2}} := \sum_{q=0}^Q \left(\binom{p-2}{q}y_j^{p-2-q} \cdot \prod_{a \in [q]} (\widehat{x_j^{(a)}}-y_j) \right).
\end{align*}
We define $\widetilde{x_j^{p-2}}$ to be the truncated Taylor series evaluated by the real value,
\[\widetilde{x^{p-2}} := \sum_{q=0}^Q \binom{p-2}{q}y_j^{p-2-q} \cdot (x_j-y_j)^q. \]
Then, we have
\begin{align*}
\widehat{x_j^{p-2}} &= \sum_{q=0}^Q \binom{p-2}{q}y_j^{p-2-q} \cdot \left((x_j-y_j)^q  \pm \frac{y_j^q}{100p^q\log^q(n)} \right) \\
&= \widetilde{x^{p-2}} \pm y_j^{p-2} \cdot \sum_{q=0}^Q \left|\binom{p-2}{q}\right| \cdot \frac{1}{100p^q\log^q(n)}.
\end{align*}
Notice that $|\binom{p-2}{q}| \le p^q$ and $Q = \O{\log n}$, then we have
\[\widehat{x_j^{p-2}} = \widetilde{x_j^{p-2}} \cdot \left(1 \pm \frac{1}{100\log(n)}\right).\]
Then, we have
\[|\widehat{x_j^{p-2}} - x_j^{p-2}| \le |\widehat{x_j^{p-2}} - \widetilde{x_j^{p-2}}| + |\widetilde{x_j^{p-2}} - x_j^{p-2}| \le x_j^{p-2} \cdot \frac{1}{10 \log(n)}.\]
where the second step follows by \lemref{lem:truncated_taylor_sampler}. 
Now, let $\calE_2$ be the event that $\AMS$ and $\FpEst$ return $2$-approximations of $\|x\|_2^2$ and $\|x\|_p^p$ respectively. 
Conditioned on $\calE_1$ and $\calE_2$, for our choice of $N' = n^{1-2/p} \cdot \polylog(n) $ we have that 
\[\frac{\widehat{F_2}}{4N'\cdot\widehat{F_p}}\cdot\widehat{x_j^{p-2}}\le\frac{x_j^{p-2}\cdot\|x\|_2^2}{n^{1-2/p}\cdot\|x\|_p^p}\le 1.\]
Therefore, our rejection probability is smaller than $1$, so it is well-defined. Moreover, by \lemref{lem:truncated_taylor_sampler} we have $\Ex{\widehat{x^{p-2}_j}}=\widetilde{x^{p-2}} = x^{p-2}_j \cdot \left(1\pm \frac{1}{\poly(n)}\right)$. Notice that $\widehat{x^{p-2}_j}$ has the same sign as $x^{p-2}_j$ conditioned on $\calE_1$, so $\Ex{\left|\widehat{x^{p-2}_j}\right|}= | x^{p-2}_j| \cdot \left(1\pm \frac{1}{\poly(n)}\right)$.
Thus, in expectation, the probability that $j$ is returned is
\[\frac{|x_j|^{p}}{4N'\cdot\|x\|_p^p}\pm\frac{1}{\poly(n)}.\]
Hence, conditioned on some index being output, the probability that $j$ is returned is
\begin{align*}
\left(\frac{|x_j|^{p}}{4N'\cdot\|x\|_p^p}\pm\frac{1}{\poly(n)}\right)&\cdot\left(\sum_{\ell\in[n]}\frac{|x_\ell|^{p}}{4N'\cdot\|x\|_p^p}\pm\frac{1}{\poly(n)}\right)^{-1}\\
&=\frac{x_j^{p}}{\|x\|_p^p}\pm\frac{1}{\poly(n)},
\end{align*}
which is the correct probability distribution. 
Finally, the probability that some index being output for each sample $i\in[N]$ is
\[\sum_{j\in[n]}\frac{|x_j|^{p}}{4N'\cdot\|x\|_p^p}\pm\frac{1}{\poly(n)}\ge\frac{1}{8N'}.\]
Thus by repeating $N = N' \cdot \O{\log n}$ times, we have that a sample is returned with probability at least $1-\frac{1}{\poly(n)}$. 
\end{proof}
Finally, we analyze the space complexity of our algorithm. 
\begin{lemma}
\lemlab{lem:frac:space}
\algref{alg:lp:sampler_frac} uses $n^{1-2/p}\cdot \polylog(n)$ bits of space.
\end{lemma}
\begin{proof}
For each $L_2$ sampler, we have its failure probability $\delta_1 $ and $\delta_2$ to be $\frac{1}{\poly(n)}$, so that by a union bound, all $N$ $L_2$ samplers succeed with probability $1 - \frac{1}{\poly(n)}$. Since we only require a constant-fractional approximation of the frequency of each $L_2$ sample $x_i$, by \thmref{thm:lp_sampler_p_small}, we use $N \cdot \O{\log^4 n} = \O{n^{1-2/p} \cdot \polylog(n)}$ to acquire the $L_2$ samples. Notice that the space consumption of AMS, \textsc{FpEst} and the $\O{\polylog(n)}$ instances of \textsc{CountSketch} is dominated by the $L_2$ samplers. Thus, we have the stated space complexity.
\end{proof}

Putting together \lemref{lem:frac:correct} and \lemref{lem:frac:space}, we have:
\begin{theorem}
\thmlab{thm:lp_sampler_frac}
Given $p>2$, there exists a perfect $L_p$ sampler that uses $\tO{n^{1-2/p}}$ bits of space. In addition, it gives an $(1+\eps)$-estimate to the sampled item using extra $\tO{\eps^{-2}n^{1-2/p}}$ bits of space.
\end{theorem}
\begin{proof}
By \thmref{thm:lp_sampler_p_small}, the $L_2$ sampler gives an $(1+\eps)$-estimate to the sampled item with high probability using $\O{\eps^{-2} \log^4 n}$ bits of space. Our space bound follows from the fact that we have $\tO{n^{1-2/p}}$ $L_2$ samplers.
\end{proof}

\subsection{Perfect Polynomial Sampler}
Now, we further generalize our results from perfect $L_p$ samplers to the following notion of perfect polynomial samplers:
\begin{definition}[Perfect polynomial sampler]
Let polynomial $G(z) = \sum_{d \in [D]} \alpha_d |z|^{p_d}$ with $0<\alpha_d<M$ for all $d \in [D]$, where $M$ and $D$ are some fixed constants. 
Given a vector $x \in \mathbb{R}^n$, a perfect polynomial sampler reports an index $i^* \in [n]$ such that for each $i \in [n]$, we have
\[\PPr{i^* = i} = \frac{G(x_i)}{\sum_{j \in [n]} G(x_j)} + \frac{1}{\poly(n)}.\]
Note that parameters $D$ and $M$ are considered constants in our setting.
\end{definition}

Similar to the above approach, we acquire an unbiased estimate to $\frac{G(x_i)}{x_i^p}$ for an index $x_i$ acquired from perfect $L_p$ sampling. 
We can then choose to accept the sampled index with a probability that must be well-defined, regardless of the sampled index $i\in[n]$. 
Our full algorithm appears in \algref{alg:poly:sampler}.

\begin{algorithm}[!htb]
\caption{Perfect polynomial sampler for polynomial with degree at most $p$}
\alglab{alg:poly:sampler}
\begin{algorithmic}[1]
\Require{Input vector $x\in\mathbb{R}^n$ in a stream, polynomial $G_p(u) = \sum_{d \in [D]} \alpha_d |u^{p_d}|$ with $0 < p_1 < p_2 < \ldots < p_d = p$ and $0<\alpha_d<M$ for all $d \in [D]$}
\Ensure{Perfect $p$-polynomial sample}
\State{Let $N = \O{\log n}$}
\State{$s_1,\ldots,s_N$ be $N$ perfect $L_p$ samples from $x$} \Comment{See \thmref{thm:lp_sampler_p_small} and \thmref{thm:lp_sampler_frac}}
\For{each $i\in[N]$}
\State{Let $j$ be the index of $s_i$}
\For{each $d \in [D]$}
\State{Let $\widehat{x_j^{p_d-p}}$ be an unbiased $(1 + \frac{1}{\log(n)})$-estimate of $x_j^{p_d-p}$} \Comment{See \thmref{thm:lp_sampler_frac}}
\EndFor
\State{\Return $j$ with probability $\frac{1}{5DM}\cdot \sum_{d \in [D]} \alpha_d \left|\widehat{x_j^{p_d-p}}\right|$}
\State{Otherwise, continue to next $i$}
\EndFor
\end{algorithmic}
\end{algorithm}

We show that our algorithm outputs a sample according to the correct probability distribution. 
\begin{lemma}
\lemlab{lem:poly:sampler:correct}
With high probability, \algref{alg:poly:sampler} outputs a perfect polynomial sample.
\end{lemma}
\begin{proof}
Note that the index $j$ is returned with probability
\[\left(\frac{|x_j|^p}{\|x\|_p^p}\pm\frac{1}{\poly(n)}\right)\cdot\frac{1}{5DM}\cdot\sum_{d \in [D]} \alpha_d \left|\widehat{x_j^{p_d-p}}\right|.\]
Recall that in \thmref{thm:lp_sampler_frac} we show our estimation to each $x_j^{p_d-p}$ gives a $(1+\frac{1}{\log(n)})$-approximation for each $d \in [D]$
with probability $1 - \frac{1}{\poly(n)}$. 
Conditioned on this event, we have that 
\begin{align*}
\frac{1}{5DM}\cdot\sum_{d \in [D]} \alpha_d \left|\widehat{x_j^{p_d-p}}\right| \le\frac{\sum_{d \in [D]} \alpha_d |x_j|^{p_d-p}}{DM} 
\le 1.
\end{align*}
Therefore, our rejection probability is smaller than $1$, so it is well-defined. Moreover, by \lemref{lem:truncated_taylor_sampler} we have $\Ex{\widehat{x^{p_d-p}_j}}= x^{p_d-p}_j \cdot \left(1\pm \frac{1}{\poly(n)}\right)$. 
Notice that $\widehat{x^{p_d-p}_j}$ has the same sign as $x^{p_d-p}_j$, so $\Ex{\left|\widehat{x^{p_d-p}_j}\right|}= | x^{p_d-p}_j| \cdot \left(1\pm \frac{1}{\poly(n)}\right)$.
Thus, in expectation, the probability that $j$ is returned is
\[\frac{G_p(x_j)}{5DM \cdot\|x\|_p^p}\pm\frac{1}{\poly(n)}.\]
Hence, conditioned on some index being output, the probability that $j$ is returned is
\begin{align*}
\left(\frac{G_p(x_j)}{5DM\cdot \|x\|_p^p}\pm\frac{1}{\poly(n)}\right)&\cdot\left(\sum_{\ell\in[n]}\frac{G_p(x_l)}{5DM\cdot \|x\|_p^p}\pm\frac{1}{\poly(n)}\right)^{-1}\\
&=\frac{G_p(x_j)}{\sum_{l \in [n]} G_p(x_l)}\pm\frac{1}{\poly(n)},
\end{align*}
which is the correct probability distribution. 
Finally, the probability that some index being output for each sample $i\in[N]$ is
\[\sum_{j\in[n]}\frac{G_p(x_j)}{5DM\cdot \|x\|_p^p}\pm\frac{1}{\poly(n)}\ge \frac{\alpha_{D} \cdot \|x\|_{p}^{p}}{5DM\cdot \|x\|_p^p} = \Omega(1).\]
Thus by repeating $N = \O{\log n}$ times, we have that a sample is returned with probability at least $1-\frac{1}{\poly(n)}$. 
\end{proof}
We next bound the space used by our polynomial sampler. 
\begin{lemma}
\lemlab{lem:poly:sampler:space}
\algref{alg:poly:sampler} uses $ n^{\max \{0, 1-2/p\}}\cdot \polylog(n)$ bits of space.
\end{lemma}
\begin{proof}
By \thmref{thm:lp_sampler_p_small} and \thmref{thm:lp_sampler_frac}, we use $n^{\max \{0, 1-2/p\}}\cdot \polylog(n)$ bits of space to acquire each $L_p$ sampler with high probability. This takes $n^{\max \{0, 1-2/p\}}\cdot \polylog(n)$ bits in total since we draw $ \polylog(n)$ samples. In addition, by \thmref{thm:lp_sampler_frac} we use $\polylog(n)$ bits of space to estimate each $x_j^{p_d-p}$. Thus, we have the stated space complexity.
\end{proof}
Putting together \lemref{lem:poly:sampler:correct} and \lemref{lem:poly:sampler:space}, we have the following full guarantees of our polynomial sampler. 
\begin{theorem}[Perfect polynomial sampler] 
\thmlab{thm:poly:sampler}
There exists a perfect polynomial sampler for a polynomial of degree at most $p$, which uses $\tO{n^{1-2/p}}$ bits of space. In addition, it gives an $(1+\eps)$-estimate $\widehat{x_i}$ to the sampled item using extra $\tO{\eps^{-2}n^{1-2/p}}$ bits of space.
\end{theorem}

\section{Approximate \texorpdfstring{$L_p$}{Lp} Sampler for \texorpdfstring{$p>2$}{p>2} with Fast Update Time} 
\seclab{sec:approx}

In this section, we present an approximate $L_p$ sampler with optimal space dependence on $n$ and $\log n$, which achieves fast update time. 
We give our algorithm in \algref{alg:lp:approx:sampler:fast}. 

\begin{algorithm}[!htb]
\caption{Approximate $L_p$ Sampler for $p>2$ with fast-update}
\alglab{alg:lp:approx:sampler:fast}
\begin{algorithmic}[1]
\Require{Input vector $x\in\mathbb{R}^n$ in a stream, accuracy $\eps$, discretization factor $\eta$}
\Ensure{Approximate $L_p$ sampler with accuracy $\eps$}
\State{Let $c\gets\Theta(1)$ be sufficiently large}
\State{For each $i\in[n],j\in[n^c]$, generate exponential random variables $\be_{i,j}^{1/p}$}
\State{For each $i\in[n]$, let $\bv_i=\max_{j\in[n^c]}|x_i| \cdot \rnd_{\eta}(1/\be_{i,j}^{1/p})$}
\State{Let $\bu\in\mathbb{R}^{n^{c+1}}$ be vector consisting of $|x_i| \cdot \rnd_{\eta}(1/\be_{i,j})$ for all $i\in[n]$, $j\in[n^c]$}
\State{Let $\bar{\bu}$ be $\bu$ with the entries of $\bv$ zeroed out}
\State{Keep a $\CountSketch_1$ with $\Theta(\log n)$ rows and $n^{1-2/p}\cdot\log \frac{1}{\eps}$ buckets on $\bv$}
\State{Let $v\in\mathbb{R}^n$ be the estimated frequencies by the resulting $\CountSketch_1$ table}
\State{For each item $v_i$ in $v$, add $v_i$ to set $B$ if its absolute value is bigger than $\frac{n^{c/p}\|x\|_p}{200\log \frac{1}{\eps}}$}
\State{Return FAIL if $B$ is empty}
\State{Keep a $\CountSketch_2$ with $\O{\log n}$ rows and the first $|B|$ of $(n^{c+1})^{1-2/p}$ buckets on $\bar{\bu}$}
\State{Add the entries in set $B$ of $\CountSketch_1$ to $\CountSketch_2$}
\State{Let $y\in\mathbb{R}^n$ be the estimated frequencies by the resulting $\CountSketch_2$ table}
\State{Let $i^*=\argmax_{i\in[n]}|y_i|$}
\State{Let $R_\bu$ be an estimation of $\|\bu\|_2$ satisfying $R \in [\frac{\|\bu\|_2}{2}, 2\|\bu\|_2]$}
\State{Let $\mu \in [\frac{1}{2}, \frac{3}{2}]$ be a uniform random variable}
\State{Return $y_{i^*}$ if $|y_{D(1)}| - |y_{D(2)}| >\frac{100 R}{\mu n^{(c+1)(1/2-1/p)}}$. Otherwise, return FAIL}
\end{algorithmic}
\end{algorithm}

As we mentioned before, the existing implementation of the perfect $L_p$ sampler requires duplicating each coordinate a polynomial number of times. This is because of the adversarial error we encounter when we condition on some specific index achieving the max. In the implementation, we fail the sampler if we do not witness a large gap between our estimate of the max and second max, since we cannot distinguish them from our CountSketch estimation with additive error. However, the failure probability may change drastically if we conditioned on different indices achieving the max. For example, consider vector $x=(100n, 1, \ldots,1) \in \mathbb{R}^n$. We would expect the first index to be the maximum in the scaled vector. And if we condition on the second index achieving the max, we would expect the max and the second max to not have a large gap. Then, the failure probability shifts by an additive constant which leads to an incorrect sampling distribution. Therefore, we duplicate each entry polynomial times so that there are no heavy-hitters in the resulting vector. This would reduce the dependency on the anti-ranks to a negligible small value. Below, we state the formal definition of the duplication vector.
\begin{definition}[Duplication]
Let $X$ denote the duplication vector where $X_{i,j} := x_i$ for all $i \in [n], j \in [n^c]$. Notice that $\|X\|_p^p = n^c \cdot \|x\|_p^p$.
\end{definition}

It is challenging to adapt this duplication method to the $p>2$ setting since we are not able to maintain a CountSketch table with $n^{(c+1)^{1-2/p}}$ buckets. 
On the other hand, we cannot implement CountSketch with a lower number of buckets or the additive error would be too high. 
Therefore, we introduce a simulation scheme by keeping a two-stage CountSketch table. 
Consider an index $i$, we generate $n^c$ exponential variables $\be_{i,j}$, we observe the maximum of the scaled vector can only be obtained from all $x_i / \be_{i,j^*}^{1/p}$, where $j^* = \argmax_{j \in [n^c]} x_i / \be_{i,j}^{1/p}$. 

Thus, in our first-stage CountSketch, we only record the maximum item $\bw_i = x_i / \be_{i,j^*}^{1/p}$ of each entry. We select a set $B$ with roughly $\polylog \frac{1}{\eps}$ indices that contains the maximum of $\bw$ with probability $1-\eps$. 
Then, in our second stage, we maintain a CountSketch table for the scaled vector $\bz$ with $\bw$ zeroed out. We only record the first $|B|$ buckets, and we discard everything that hashes into other buckets. 
We add up the entries in the two CountSketch tables to obtain an estimate for each item in $B$. Then, we do a statistical test to fail the instances that do not have anti-concentration. 
Notice that since we hash the items uniformly randomly, this simulation will not change the distribution of our estimation. In this way, we implement the approximate sampler with a small space occupation.

However, upon each arrival in the stream, if we calculate each of the $n^c$ duplicated scaled entries explicitly in the above simulation method, the update time would be $\O{n^c}$, which is prohibitively large in practice.
To speed up the update time, instead of calculating the explicit value of $\frac{X_i}{\be_i^{1/p}}$, we scale $X_i$ by $\rnd_\eta(1/\be^{1/p})$, where $\rnd_\eta(x)$ rounds $x$ down to the nearest power of $(1+\eta)^q$ for $q \in \mathbb{Z}$. Therefore, we maintain a different vector $\bu \in \mathbb{R}^{n^{(c+1)}}$ in our CountSketch table, where $\bu_i = X_i \cdot \rnd_\eta(1/\be^{1/p})$. Notice that $\bu_i = \bz_i \cdot (1 \pm \O{\eta})$ for all $i \in [n^{c+1}]$, where $\bz_i = X_i / \be^{1/p}$ is the scaled entry without round-up. Thus, this gives us an approximate $L_p$ sampler with accuracy $\O{\eta}$. 

Last, we state a modification to the classic CountSketch algorithm introduced by \cite{JayaramW18}, which is used to reduce the dependency on the anti-ranks. Let $A \in \mathbb{R}^{d \times l}$ be a $d \times l$ CountSketch matrix. Instead of uniformly hashing each item into a bucket in each row, we generate variables $h_{i, j, k} \in\{0,1\}$ for $(i, j, k) \in[d] \times[l] \times[n]$, where $h_{i, j, k}$ are all i.i.d. and equal to 1 with probability $1 / l$. 
We also let $g_{i, k} \in\{1,-1\}$ be i.i.d. Rademacher variables (1 with probability $1 / 2$ ). Then $A_{i, j}=\sum_{k=1}^n x_k g_{i, k} h_{i, j, k}$, and the estimate $y_l$ of $x_k$ is given by:
\[y_k=\operatorname{median}\left\{g_{i, k} A_{i, j} \mid h_{i, j, k}=1\right\}\]
Thus, it performs the same as hashing the item to $k$ random buckets in the whole $\CountSketch$ table. Note the element $f_k$ can be hashed into multiple buckets in the same row of $A$, or even be hashed into none of the buckets in a given row. We remark that the error bound of the original $\CountSketch$ algorithm can be applied as usual (see Section A.1 of \cite{JayaramW18} for detailed analysis).

\paragraph{Analysis of $\CountSketch_1$.}
We recall that in $\CountSketch_1$, for each index $k \in [n]$, we select the maximum $1/\be_{k,j^*}$ of $1/\be_{k,j}, j \in [n^c]$. Then, we implement countsketch on the resulting scaled vector $\bv_k = x_k/\be_{k,j^*}$.
Now, we show that the $\CountSketch_1$ table recovers the maximum index of the scaled vector with high probability.

First, we state Bernstein's inequality.
\begin{theorem}[Bernstein's inequality]
\thmlab{thm:bernstein}
Let $Y_1,\ldots,Y_n$ be independent random variables such that $|Y_i|\le M$ for all $i\in[n]$, and $\text{Var}\left(\sum_i Y_i\right)\le\sigma^2$. 
Then there exist constants $C_1,C_2>0$ such that for all $t>0$,
\[\PPr{\sum_i Y_i-\Ex{\sum_i Y_i}>t}\le C_1\left(e^{-C_2t^2/\sigma^2}+e^{-C_2t/M}\right).\]
\end{theorem}

Next, we introduce a lemma that bounds the number of large items in the scaled vector.

\begin{lemma}
\lemlab{lem:large:index}
Let $C$ be a constant. We call an index $k\in[n]$ \emph{large} if $\bv_k\ge\frac{n^{c/p}\|x\|_p}{C\log \frac{1}{\eps}}$. Let $\calE$ be the event that the number of large indices is at most $2C^p\log^{p+1} \frac{1}{\eps}$, we have $\PPr{\calE}\ge 1 - \frac{1}{\poly(n)}$.
\end{lemma}
\begin{proof}
We define $w_k = \max_{j \in [n^c]} x_k/\be_{k,j}^{1/p}$. 
By the max-stability of exponential variables, we have $w_k\sim\frac{x_k\cdot n^{c/p}}{\be^{1/p}}$ for an exponential random variable $\be$, which is independent of $\be_{i,j}$. Then, we have
\[\PPr{w_k\ge\frac{n^{c/p}\|x\|_p}{C\log \frac{1}{\eps}}}=\PPr{\be\le\frac{x_k^p}{\|x\|_p^p}\cdot C^p\log^p \frac{1}{\eps}}\le\frac{x_k^p}{\|x\|_p^p}\cdot C^p\log^p \frac{1}{\eps},\]
where the last inequality is from the cumulative density function of exponential random variables. 
Consider the discretization vector $\bv$, since $\bv_k = w_k \cdot (1 \pm \O{\eta})$ for each $k \in [n]$, the above equation still holds for $\eta < \frac{1}{10}$.

Now, for each $k \in [n]$, we define variable $Y_k$ to be $1$ if $\bv_k$ is large and $0$ otherwise. Note that $Y_k$'s are independent.
Hence, we have that $\Ex{\sum_{k \in [n]} Y_k} \le C^p\log^p \frac{1}{\eps}$ and $\Var{\sum_{k \in [n]} Y_k} \le \sum_{k \in [n]} \Ex{Y_k^2} \le C^p\log^p \frac{1}{\eps}$.
Let $\calE$ be the event that the number of large indices is at most $2C^p\log^{p+1} \frac{1}{\eps}$. Thus, by Bernstein's inequality (c.f., \thmref{thm:bernstein}), we have $\PPr{\calE}\ge 1 - \frac{1}{\poly(n)}$. 
\end{proof}

The following lemma upper bounds the error of the first CountSketch table.
\begin{lemma}
\lemlab{lem:cs:err:w:descrete}
With probability at least $1-\frac{1}{\poly(n)}$, the error of $\CountSketch_1$ with $n^{1-2/p}\cdot\log \frac{1}{\eps}$ buckets and $\Theta(\log n)$ rows is at most $\frac{n^{c/p}\|x\|_p}{400\log \frac{1}{\eps}}$.  
\end{lemma}
\begin{proof}
Let $L=n^{1-2/p}\cdot\log \frac{1}{\eps}$ be the number of buckets in the $\CountSketch_1$ table. 
Consider an index $k \in [n]$, recall that we generate a hash variable $h_{i,j,k}$ which is $1$ with probability $\frac{1}{L}$ for each bucket $(i,j)$, and we hash $k$ to a bucket if $h_{i,j,k} = 1$. Then, since we have $d = \Theta(\log n)$ rows, $k$ is hashed to $\Theta(\log n)$ buckets in the $\CountSketch_1$ table with probability $1-\frac{1}{\poly(n)}$. We use $\calE_1$ to denote the event that each index $k \in [n]$ is hashed to $\Theta(\log n)$ buckets, and $\calE_1$ happens with high probability by a union bound. We condition on event $\calE_1$ in the following analysis.

We call an index $k\in[n]$ \emph{large} if $\bv_k\ge\frac{n^{c/p}\|x\|_p}{C\log \frac{1}{\eps}}$. We define $\calE_2$ to be the event that the number of large indices is at most $2C^p\log^{p+1} \frac{1}{\eps}$. By \lemref{lem:large:index}, $\calE_2$ happens with probability $1 - \frac{1}{\poly(n)}$.

For a fixed bucket $(i,j) \in [n^{1-2/p} \cdot \log \frac{1}{\eps}] \times [\Theta(\log n)]$. Let $\calE_3$ be the event that none of the large indices are hashed to $(i,j)$, so conditioned on $\calE_1$ and $\calE_2$, we have $\PPr{\calE_3}\ge 1 - \O{\frac{\log^{p+1} \frac{1}{\eps}}{n^{1-2/p}}} \ge 0.99$. 

For a fixed index $k$ such that $h_{i,j,k} = 1$, the error in the estimate of $\bv_k$ given by bucket $(i,j)$ is $S_{i,j}=\sum_{r\in[n],r\neq k}I(h_{i,j,r}=1)\cdot g_{i,r}\cdot \bv_r$, where $I(h_{i,j,r}=1)$ is the indicator function so that $I(h_{i,j,r}=1)=1$ if $h_{i,j,r}=1$ and $I(h_{i,j,r}=1)=0$ otherwise, and $g_{i,k}\in\{\pm1\}$ is the random sign corresponding to $k$ in row $i$ of $\CountSketch_1$. 
Then we have $\Ex{S_{i,j}}=0$ and from the analysis in \lemref{lem:z_second_moment}, $\Ex{S_{i,j}^2}\le\frac{1}{n^{1-2/p} \log \frac{1}{\eps}}\cdot\O{n^{2c/p}}\cdot\|x\|_2^2 \le \O{\frac{n^{2c/p}\cdot\|x\|_p^2}{\log \frac{1}{\eps}}}$. 
Conditioned on $\calE_3$, we have $|\bv_r|<\frac{n^{c/p}\|x\|_p}{C\log \frac{1}{\eps}}$ for all $r\neq k$ such that $h_{i,j,r} = 1$. 
Now by Bernstein's inequality, c.f., \thmref{thm:bernstein} for $t=\frac{n^{c/p}\|x\|_p}{200\log \frac{1}{\eps}}$ and $L=n^{1-2/p}\cdot\log \frac{1}{\eps}$, we have
\[\PPr{|S_{i,j}|>\frac{n^{c/p}\|x\|_p}{400\log \frac{1}{\eps}}\,\mid\,\calE_1\wedge\calE_2}\le\frac{1}{100}.\]
Thus by a union bound, we have that with probability at least $0.97$, the estimate of $k$ in bucket $(i,j)$ of $\CountSketch_1$ has an additive error at most $\frac{n^{c/p}\|x\|_p}{400\log \frac{1}{\eps}}$. 
By taking the median across $\Theta(\log n)$ buckets and a union bound across all $k\in[n]$, we have that the estimate for $\CountSketch_1$ for each index $k\in[n]$ is at most $\frac{n^{c/p}\|x\|_p}{400\log \frac{1}{\eps}}$ with high probability. 
\end{proof}

The following lemma shows the maximum entry in the scaled vector is $\frac{1}{\log \frac{1}{\eps}}$-heavy compared to the $p$-th norm of the unscaled vector.
\begin{lemma}
\lemlab{lem:countsketch_max_w}
We have $\max_{i\in[n]} |\bv_i|\ge\frac{n^{c/p}\|x\|_p}{100\log  \frac{1}{\eps}}$ with probability $1-\poly(\eps)$. 
\end{lemma}
\begin{proof}
Notice that $\max_{i\in[n]} |w_i| = \max_{i \in [n], j \in [n^c]} |\bz_{i,j}|$. Then, by \lemref{lem:max_stability}, the following holds with probability $1-\poly(\eps)$:
\[\max_{i\in[n]}| w_i| \ge\frac{1}{110\log \frac{1}{\eps}}\cdot\|X\|_p = \frac{1}{110\log \frac{1}{\eps}}\cdot n^{c/p}\cdot\|x\|_p.\]
Consider the discretization vector $\bv_k = w_k \cdot (1\pm \O{\eta})$, the claim holds for $\eta < \frac{1}{10}$.
\end{proof}

Combining  \lemref{lem:cs:err:w:descrete} and \lemref{lem:countsketch_max_w}, we have the following result showing the correctness of $\CountSketch_1$.
\begin{lemma}
\lemlab{lem:cs:discrete}
We recover the maximum index $i^* = \argmax \bv_i$ in $\CountSketch_1$ with probability $1 - \poly(\eps) - \frac{1}{\poly(n)}$.
\end{lemma}


\paragraph{Analysis of $\CountSketch_2$.}
We prove the correctness of $\CountSketch_2$. As we mentioned in \secref{sec:approx}, $\CountSketch_2$ fails the statistical test with constant probability, we need to bound the dependency of this probability on which index achieves the maximum.

We introduce the following corollary of the Khintchine inequality, which is used to bound the error of each CountSketch bucket.

\begin{lemma}
\lemlab{lem:cs_bounded_est}
\label{lem:cs_bounded_est}
Let $r_1,\ldots,r_n\in\{-1,+1\}$ be independent random signs. 
Then for any $c>0$, there exists a constant $C$ such that 
\[|r_1x_1+\ldots+r_nx_n|\le C\cdot\sqrt{\log n}\cdot\|x\|_2.\]
with probability at least $1-\frac{1}{n^c}$. 
\end{lemma}
\begin{proof}
By \thmref{thm:khintchine} and \propref{prop:gamma}, we have that
\[\Ex{|r_1x_1+\ldots+r_nx_n|^p}\le(B\cdot\sqrt{p})^p\cdot\|x\|_2^p,\]
for some absolute constant $B>0$. 
Then for any constant $c>0$, we have by Markov's inequality,
\[\PPr{|r_1x_1+\ldots+r_nx_n|^p\ge n^c\cdot(B\cdot\sqrt{p})^p\cdot\|x\|_2^p}\le\frac{1}{n^c}.\]
Thus for $p=\log(n)$, we have
\[\PPr{|r_1x_1+\ldots+r_nx_n|\ge (2^c)\cdot(B\cdot\sqrt{\log n})\cdot\|x\|_2}\le\frac{1}{n^c}.\]
\end{proof}

The next lemma shows that the dependency on the anti-ranks is reduced by the duplication.
\begin{lemma}\cite{JayaramW18} 
\lemlab{lem:bz_independence}
Let $N$ be the cardinality of the support of the duplicated vector $X$. Let $\bz$ be the scaled duplication vector. For every $1 \leq k<N-n^{9 (c+1) / 10}$, with probability $1-\O{e^{-n^{(c + 1)/ 3}}}$ we have
$$ \left|\bz_{D(k)}\right|=\left[\left(1 \pm \O{n^{- (c+1) / 10}} \right) \sum_{\tau=1}^k \frac{E_\tau}{\mathbb{E}\left[\sum_{j=\tau}^N\left|X_{D(j)}\right|^p\right]}\right]^{-1 / p}, $$
where $(E_1,\cdots, E_N)$ are i.i.d. exponential variables with rate $1$ which are independent of the anti-rank vector $D = (D(1),\cdots,D(N))$.   
\end{lemma}

The next lemma shows that if we decompose a random variable $Z$ as the sum of a variable $A$ which is independent of an event $I$ and a variable $B$ that depends on $I$ but has a small value. Then, we can bound the linear dependency of $Z$ on $I$.
\begin{lemma}
\lemlab{lem:sum_independent}
\cite{JayaramW18} Let $A, B \in \mathbb{R}^d$ be random variables where $Z=A+B$. Suppose $A$ is independent of some event $E$, and let $M>0$ be such that for every $i \in[d]$ and every $a<b$ we have $\operatorname{Pr}[a \leq$ $\left.A_i \leq b\right] \leq M(b-a)$. Suppose further that $|B|_{\infty} \leq \eps$. Then if $I=I_1 \times I_2 \times \cdots \times I_d \subset \mathbb{R}^n$, where each $I_j=\left[a_j, b_j\right] \subset \mathbb{R},-\infty \leq a_j<b_j \leq \infty$ is a (possibly unbounded) interval, then
$$\operatorname{Pr}[Z \in I \mid E]=\operatorname{Pr}[Z \in I]+\O{\eps d M}.$$
\end{lemma}

Next, we bound the linear dependency of $\CountSketch_2$ on the anti-ranks. We remark that the following proof is based on the analysis of Lemma 12 in \cite{JayaramW18}.

\begin{lemma}
\lemlab{lem:indep:approx:discrete}
Let $\neg \mathrm{FAIL}$ denote the event that $\CountSketch_2$ does not fail. Suppose $\eta > n^{-c}$, we have $\PPr{\neg \mathrm{FAIL} ~|~ D(1)} = \PPr{\neg \mathrm{FAIL}} \pm \O{\eta\sqrt{\log n}}$.
\end{lemma}
\begin{proof}
The idea of the proof is decomposing the estimation of $\CountSketch_2$ as the sum of an independent variable and a dependent variable with a small magnitude, so that we can apply \lemref{lem:sum_independent}. We start by decomposing each entry in the scaled vector $\bv$.
\paragraph{Decomposition of each entry of $\bv$.}
Conditioned on \lemref{lem:bz_independence} holding, for every $k<N-n^{9 (c+1) / 10}$ we have 
\[\left|\bz_{D(k)}\right|=U_{D(k)}^{1 / p}\left(1 \pm \O{n^{-(c+1) / 10}}\right)^{1 / p}=U_{D(k)}^{1 / p}\left(1 \pm \O{\frac{1}{p} n^{-(c+1) / 10}}\right),\]
where $U_{D(k)}=\left(\sum_{\tau=1}^k \frac{E_\tau}{\mathbb{E}\left[\sum_{j=\tau}^N\left|F_{D(j)}\right|^p\right]}\right)^{-1}$ is totally determined by $k$ and the hidden exponentials $E_i$, and thus, independent of the anti-rank vector $D$. Thus, we can decompose $\bz_{D(k)}$ as
\[\left|\bz_{D(k)}\right| = \calU_{D(k)}^{1 / p}+\calU_{D(k)}^{1 / p} \calV_{D(k)},\]
where $\calU_{D(k)}$ is independent of the anti-ranks and $\calV_{D(k)}$ is some random variable that satisfies $\left|\calV_{D(k)}\right|=\O{n^{-(c+1)/10}}$. Note that we round-up vector $\bz$ to the nearest power $(1+\eta)^q$ to retrieve $\bv$, so we can decompose $\bv$ as follows,
\[\left|\bz_{D(k)}\right| = U_{D(k)}^{1 / p}+U_{D(k)}^{1 / p} V_{D(k)},\]
where $U_{D(k)}$ is independent of the anti-ranks and $V_{D(k)}$ is some random variable that satisfies $\left|V_{D(k)}\right|=\O{\eta}$.

\paragraph{Decomposition of the $\CountSketch$ table.}
We consider a bucket $A_{i, j}$ in the $\CountSketch$ table for $(i, j) \in[d] \times\left[l\right]$, where $d = \Theta(\log n)$ is the number of rows and $l = n^{c^{1-2/p}}$ is the number of buckets in each row. Let $\sigma_k=\operatorname{sign}\left(\bw_k\right)$ for $k \in\left[n^{c+1}\right]$. Then we have
\[A_{i, j}= \sum_{k \in B_{i j}} \sigma_{D(k)}\left|\bv_{D(k)}\right| g_{i, D(k)}+\sum_{k \in S_{i j}} \sigma_{D(k)}\left|\bv_{D(k)}\right| g_{i, D(k)},\]
where $B_{i j}=\left\{k \leq N-n^{9 (c+1) / 10} \mid h_{i, j, D(k)}=1\right\}$ and $S_{i j}=\left\{n^{c+1} \geq k>N-n^{9 (c+1) / 10} \mid h_{i, j, D(k)}=1\right\}$. Here we define $\left\{D(N+1), \ldots, D\left(n^{c+1}\right)\right\}$ to be the set of indices $i$ with $X_i=0$. So
\[A_{i, j}=\sum_{k \in B_{i j}} g_{i, D(k)} \sigma_{D(k)} U_{D(k)}^{1 / p}+\sum_{k \in B_{i j}} g_{i, D(k)} \sigma_{D(k)} U_{D(k)}^{1 / p} V_{D(k)}+\sum_{k \in S_{i j}} g_{i, D(k)} \bv_{D(k)}.\]
We upper bound the last two terms in the next lemma.
\begin{lemma}
\lemlab{lem:bucket:bound:discrete}
For all $i, j$, we have 
\[\left|\sum_{k \in B_{i j}} g_{i, D(k)} \sigma_{D(k)} U_{D(k)}^{1 / p} V_{D(k)}\right|+\left|\sum_{k \in S_{i j}} g_{i, D(k)} \bv_{D(k)}\right|\le \O{ \eta\sqrt{\log n} \cdot \|\bv\|_2},\]
with probability $1-\O{\frac{1}{\poly(n)}}$.
\end{lemma}
\begin{proof}
For the first term, we have $|V_{D(k)}| = \O{\eta}$. Then, by \lemref{lem:cs_bounded_est}, we have
\[\left|\sum_{k \in B_{i j}} g_{i, D(k)} \sigma_{D(k)} U_{D(k)}^{1 / p} V_{D(k)}\right| \le \O{\eta}\cdot \sqrt{\log n} \cdot \|\bz\|_2,\]
with probability $1 - n^{-c}$.
For the second term, again by \lemref{lem:cs_bounded_est}, we have
\[\left|\sum_{k \in S_{i j}} g_{i, D(k)} \bv_{D(k)}\right| = \O{\sqrt{\log n} \cdot  \| \bz_{S_{ij}}\|_2^2},\]
with probability $1- n^{-c}$, where $r_{D(k)}$'s are the random sign vectors assign to $D(k)$ in the \textsc{CountSketch}. Since $S$ consists of the $n^{9(c+1)/10}$ smallest non-zero coordinates of $\bz$ we have
\[\left|\sum_{k \in S_{i j}} g_{i, D(k)}\bv_{D(k)}\right| =\O{ \eta\sqrt{\log n} \cdot \|\bz\|_2}, \]
with probability $1- n^{-c}$. 
\end{proof}
Conditioned on \lemref{lem:bucket:bound:discrete} holds, we can decompose $\left|A_{i, j}\right|$ into $\left|\sum_{k \in B_{i j}} g_{i, D(k)} \sigma_{D(k)} U_{D(k)}^{1 / p}\right|+V_{i, j}$ where $V_{i, j}$ is some random variable satisfying $\left|V_{i, j}\right|=\O{ \eta\sqrt{\log n} \cdot \|\bz\|_2}$. Let $U_{i, j}^*=\left|\sum_{k \in B_{i j}} g_{i, D(k)} \sigma_{D(k)} U_{D(k)}^{1 / p}\right|$. Let $\Gamma(k)=\left\{(i, j) \in[d] \times[l] \mid h_{i, j, D(k)}=1\right\}$. Then our estimate for $\left|\bz_{D(k)}\right|$ is $$y_{D(k)}=\operatorname{median}_{(i, j) \in \Gamma(l)}\left\{U_{i, j}^*+V_{i, j}\right\}=\operatorname{median}_{(i, j) \in \Gamma(l)}\left\{U_{i, j}^*\right\}+V_{D(k)}^*,$$
where $\left|V_{D(k)}^*\right|=\O{ \eta\sqrt{\log n} \cdot \|\bz\|_2}$ for all $k \in\left[n^{c+1}\right]$. Note that $\operatorname{median}_{(i, j) \in \Gamma(k)}\left\{U_{i, j}^*\right\}$ is independent of the anti-ranks.

\paragraph{Decomposition of $L_2$ estimation.}
We now consider our $L_2$ estimation, which is given by $R=\frac{5}{4} \operatorname{median}_j\left\{\left|\sum_{k \in[n^{c+1}]} \varphi_{k,j} \bv_k\right|\right\}$ where the $\varphi_{k, j}$ 's are i.i.d. normal Gaussians. We can write this as
\[R=\frac{5}{4} \operatorname{median}_j\left\{\left|\sum_{k \in B} \varphi_{D(k),j} \sigma_{D(k)} U_{D(k)}^{1 / p}+\left(\sum_{k \in B} \varphi_{D(k),j} \sigma_{D(k)} U_{D(k)}^{1 / p} V_{D(k)}+\sum_{k \in S} \varphi_{D(k),j} \bv_{D(k)}\right)\right|\right\},\]
where $B$ and $S$ are the union of all $B_{ij}$ and $S_{ij}$ respectfully.

The next lemma upper bounds the last two terms in the above formula.
\begin{lemma}
\lemlab{lem:L2:est:discrete}
For all $i, j$, we have 
\[\left|\sum_{k \in B} \varphi_{ D(k),j} \sigma_{D(k)} U_{D(k)}^{1 / p} V_{D(k)}\right|+\left|\sum_{k \in S} \varphi_{D(k),j} \bv_{D(k)}\right|= \O{ \eta\sqrt{\log n} \cdot \|\bz\|_2},\]
with probability $1-\O{\frac{1}{\poly(n)}}$.
\end{lemma}
\begin{proof}
Notice that Gaussian variables are $2$-stable, which means that for any vector $x \in \mathbb{R}^n$, if $\varphi_1, \ldots, \varphi_n$ are i.i.d. Gaussian, then 
\[\operatorname{Pr}\left[\left|\sum_i \varphi_i x_i\right|>\O{\sqrt{\log n}}\|x\|_2\right]=\operatorname{Pr}\left[|\varphi|\|x\|_2>\O{\sqrt{\log n}}\|x\|_2\right],\] where $\varphi$ is an independent Gaussian variable. Moreover, $\operatorname{Pr}\left[|\varphi|\|x\|_2>\O{\sqrt{\log n}}\|x\|_2\right] < n^{-c}$ due to the pdf of Gaussian variables. Therefore, replacing the inputs of the Khintchine inequality in the proof of \lemref{lem:bucket:bound:discrete} will give us the proof of \lemref{lem:L2:est:discrete}.
\end{proof}
Conditioned on \lemref{lem:L2:est:discrete} holds, we have $R=\frac{5}{4} \operatorname{median}_j\left\{\left(\left|\sum_{k \in B} \varphi_{D(k) j} \sigma_{D(k)} U_{D(k)}^{1 / p}\right|\right\}+V_R\right.$where the median is independent of the anti-ranks and $\left|V_R\right|=\O{ \eta\sqrt{\log n} \cdot \|\bz\|_2}$.

\paragraph{Correctness of the second criterion.}
We define $ U_{D(k)}^*=\operatorname{median}_{(i, j) \in \Gamma(k)}\left\{U_{(i, j)}^*\right\}$ and $$U_R^*=\frac{5}{4} \operatorname{median}_j\left(\left|\sum_{k \in B} \varphi_{D(k) j} \sigma_{D(k)} U_{D(k)}^{1 / p}\right|\right).$$ Then, we can decompose our $\CountSketch_2$ estimation and $L_2$ estimation as $y_{D(k)}=U_{D(k)}^*+V_{D(k)}^*$ and $R = U_R^* + V_R^*$. From \lemref{lem:bucket:bound:discrete} and \lemref{lem:L2:est:discrete}, we have both $U_{D(k)}^*, U_R^*$ are independent of the anti-ranks $D(k)$, and $|V_{D(k)}^*| + |V_R^*| = \O{ \eta\sqrt{\log n} \cdot \|\bz\|_2}$.

Now, we define a deterministic function $\Lambda(x, v)$, such that for vector $x$ and a scalar $v$, set $\Lambda(x, v)=x_{D(1)}-x_{D(2)}- v$. (Indeed, in our algorithm, $D(1)$ and $D(2)$ should be the anti-rank from the scaled vector $\bw$, which consists of the coordinates $x_i$ scaled by the max of the $n^c$ inverse exponential variables, instead of the duplication scaled vector $\bz$. However, we can consider $D(i)$ such that $\bz_{D(i)} = z_{D(2)}$ and set $\Lambda(x, v)=x_{D(1)}-x_{D(i)}- v$, which will not affect the correctness of the analysis). 
Notice that our second criterion is equivalent to $\Lambda(y, \frac{100 R}{\mu n^{(c+1)(1/2-1/p)}}) \ge 0$. Conditioned on \lemref{lem:bucket:bound:discrete} and \lemref{lem:L2:est:discrete}, we can decompose $\Lambda(y, \frac{100 U_R^*}{\mu n^{(c+1)(1/2-1/p)}})$ into
\[\Lambda(y, \frac{100 R}{\mu n^{(c+1)(1/2-1/p)}}) = \Lambda(U^*,\frac{100 U_R^*}{\mu n^{(c+1)(1/2-1/p)}}) + V,\] 
where $U^*$ and $U_R^*$ are independent of the anti-ranks, and $V$ satisfies $V = \O{ \eta\sqrt{\log n} \cdot \|\bz\|_2}$. Now for any interval $I$, since $\nu$ is a uniform random variable, we have
\[ \operatorname{Pr}\left[\Lambda\left(U^*, \frac{100 U_R^*}{\mu n^{(c+1)(1/2-1/p)}}\right) \in I\right]=\operatorname{Pr}\left[\mu \in I^{\prime} \cdot \frac{100 U_R^*}{n^{(c+1)(1/2-1/p)}}\right] 
=\O{|I| \cdot \frac{100 U_R^*}{n^{(c+1)(1/2-1/p)}}},\]
where $I^{\prime}$ is the result of shifting the interval $I$ by a term which is independent of $\mu$. Here $|I| \in[0, \infty]$ denotes the size of the interval $I$. Thus it suffices to lower bound $U_R^*$. We have $2 U_R^*>R>\frac{1}{2}\|\bz\|_2$ after conditioning on the success of our $L_2$ estimator, which holds with probability $1-n^{-c}$. Thus $\operatorname{Pr}\left[\Lambda\left(\vec{U}^*, \frac{100\mu U_R^*}{n^{(c+1)(1/2-1/p)}} \right) \in I\right]=\O{|I| /\frac{100 \|\bz\|_2}{n^{(c+1)(1/2-1/p)}}}$ for any interval $I$. So, applying \lemref{lem:sum_independent} by taking $A = \Lambda\left(U^*, \frac{100 U_R^*}{\mu n^{(c+1)(1/2-1/p)}} \right)$ and $B = V$, we have
\[\operatorname{Pr}\left[\Lambda(y,\frac{100 R}{\mu n^{(c+1)(1/2-1/p)}}) \geq 0\mid D(1)\right]=\operatorname{Pr}\left[\Lambda(y, \frac{100 R}{\mu n^{(c+1)(1/2-1/p)}}) \geq 0 \right] \pm \O{\eta \sqrt{\log n}},\]
for sufficiently large $c$. 
Note that \lemref{lem:bucket:bound:discrete} and \lemref{lem:L2:est:discrete} holds with probability $1-\frac{1}{\poly(n)}$, which completes the proof of the lemma.
\end{proof}

The next statement upper bounds the failure probability of the second-stage CountSketch.
\begin{lemma}
\lemlab{lem:fail:approx:discrete}
Let $\neg \mathrm{FAIL}$ be the event that $\CountSketch_2$ does not fail. Then $\PPr{\neg \mathrm{FAIL}} = \Omega(1)$.
\end{lemma}
\begin{proof}
First, we show there is a gap between the first max and the second max of the scaled vector with constant probability. By \lemref{lem:exp_inverse}, we have $|\bz_{D(1)}| = \|X\|_p / E_1^{1/p}$ and $|\bz_{D(2)}| = (E_1 / \|X\|_p^p + E_2 / \|F_{-D(1)}\|_p^p)^{-1/p}$, where $E_1$ and $E_2$ are independent exponential variables. Now, we have
\[|\bz_{D(2)}| = \left(\frac{E_1 }{ \|X\|_p^p} + \frac{E_2}{\|X\|_p^p \cdot(1\pm n^c)}\right)^{-1/p} = \frac{\|X\|_p}{(E_1 + E_2\cdot(1\pm n^c))^{1/p}}.\]
Therefore, due to the pdf of exponential variables, 
\[\PPr{|\bz_{D(1)}| - |\bz_{D(2)}| > \Theta(\|X\|_p)} > \Omega(1).\]
Now, consider $\bz_{D(i)} = \bw_{D(2)}$, which is the true value of our estimation $y_{(2)}$, obviously $|\bz_{D(i)}| \le |\bz_{D(2)}|$, so we have
\[\PPr{|\bz_{D(1)}| - |\bz_{D(i)}| > \Theta(\|X\|_p)} > \Omega(1).\]
Since $\bv = \bz\cdot (1\pm \O{\eta})$, the above equation still holds for $\bv$ if $\eta \le \frac{1}{10}$.

Moreover, by \lemref{lem:cs:err:w:descrete}, the error of the estimation to $\bv_i$ in $\CountSketch_1$ is at most $\O{\|X\|_p}$ with high probability. By the standard guarantee of CountSketch, the additive error due to $\CountSketch_2$ with $[(n^{c+1})^{1-2/p} ] \times\left[\O{\log n}\right]$ buckets is at most
\[\O{\frac{\|\bv\|_2}{n^{(c+1)(1/2-1/p)}}} = \O{\frac{\|\bz\|_2}{n^{(c+1)(1/2-1/p)}}} = \O{\frac{\|F\|_2}{n^{(c+1)(1/2-1/p)}}} = \O{\|X\|_p},\]
where the first step holds with probability $\Omega(1)$ by \lemref{lem:z_second_moment}. Thus, the error of the estimate $y_i$ is at most $\O{\|X\|_p}$ for all $i \in B$ with constant probability. Then, we have
\[\PPr{y_{(1)} - y_{(2)}> \Theta(\|X\|_p)} > \Omega(1).\]
Moreover, since $R$ is a $2$-approximation of $\|\bv\|_2$, similarly we have
\[ \frac{R}{n^{(c+1)(1/2-1/p)}} = \O{\frac{\|\bv\|_2}{n^{(c+1)(1/2-1/p)}}} = \O{\|X\|_p}.\]
Combining the above bounds, we have
\[ \PPr{y_{(1)} - y_{(2)} > \frac{100\mu R}{n^{(c+1)(1/2-1/p)}}} = \Omega(1).\]
This implies that $\PPr{\neg \mathrm{FAIL}} = \Omega(1)$.
\end{proof}

\paragraph{Correctness of \algref{alg:lp:approx:sampler:fast}.}
We show that our approximate sampler has the correct sampling distribution.
\begin{theorem}
\thmlab{thm:approx:discrete}
For $\eta = \frac{\O{\eps}}{\sqrt{\log n}}$, \algref{alg:lp:approx:sampler:fast} outputs an index $i \in[n]$ such that for each index $j \in [n]$, we have
\[\PPr{i=j}=\frac{\left|f_j\right|^p}{\|X\|_p^p}(1 \pm \eps) \pm\frac{1}{\poly(n)},\]
and outputs FAIL with probability at most $\O{1} < 1$.
\end{theorem}
\begin{proof}
First, we show whenever some index $i^*$ is reported, it satisfies $i^*,j^* = \argmax_{i \in [n], j \in [n^c]} |x_i/\be_{i,j}|$. Due to our statistical test, we have $y_{(1)} - y_{(2)} > \frac{100\mu R}{n^{(c+1)(1/2-1/p)}}$. Then the gap between the estimations of the top two coordinates in $\bv$ is at least $50$ times the \textsc{CountSketch} error. This means that $\bu_{(1)} $ is strictly larger than $ \bu_{(2)}$. Let the round-up factor $\eta < \frac{1}{10}$, we have $w_{(1)} $ is strictly larger than $ w_{(2)}$, and hence we output the correct max.

Then, an index $i$ is reported if and only if the following conditions are satisfied,
\begin{enumerate}
    \item $i,j$ is the maximum coordinate in the scaled vector for some $j \in [n^c]$, denoted by $\calE_{i,j}$.
    \item $i$ is recovered by $\CountSketch_1$, denoted by event $\calE^*$.
    \item $\CountSketch_2$ does not fail, denoted by event $\neg \mathrm{FAIL}$.
\end{enumerate}
Then, \algref{alg:lp:approx:sampler:fast} reports $i$ with probability
\begin{align*}
\sum_{j \in\left[n^{c}\right]} &\PPr{\calE^*, \neg \mathrm{FAIL} \mid \calE_{i,j}} \PPr{\calE_{i,j}}.
\end{align*}
By \lemref{lem:cs:discrete}, we have
\[\PPr{\calE^* \mid \calE_{i,j}} = 1 - \poly(\eps) - \frac{1}{\poly(n)}.\]
By \lemref{lem:indep:approx:discrete}, we have
\begin{align*}
\PPr{\neg \mathrm{FAIL} \mid \calE_{i,j}} = \PPr{ \neg \mathrm{FAIL} } \pm \O{\eta \cdot \sqrt{\log n}} = q \pm \O{\eta \cdot \sqrt{\log n}} .
\end{align*}
where $q = \Omega(1)$ from \lemref{lem:fail:approx:discrete}. Thus, we have
\[\PPr{\calE^*, \neg \mathrm{FAIL} \mid \calE_{i,j}}  \le \PPr{ \neg \mathrm{FAIL} \mid \calE_{i,j}} \ge q + \O{\eta \cdot \sqrt{\log n}}. \]
Moreover, by a union bound, we have
\begin{align*}\PPr{\neg \calE^*, \mathrm{FAIL} \mid \calE_{i,j}} & ~ \le \PPr{\neg \calE^* \mid \calE_{i,j}} + \PPr{ \mathrm{FAIL} \mid \calE_{i,j}} \\
& ~ \le \poly(\eps) + \frac{1}{\poly(n)} + 1 - q - \O{\eta \cdot \sqrt{\log n}}.
\end{align*}
Therefore, we have
\[\PPr{\calE^*, \neg \mathrm{FAIL} \mid \calE_{i,j}} \ge q - \O{\eta \cdot \sqrt{\log n}} - \poly(\eps) - \frac{1}{\poly(n)}.\]
Due to our choice of $\eta = \frac{\O{\eps}}{\sqrt{\log n}}$, we have
\[\PPr{\neg \calE^*, \mathrm{FAIL} \mid \calE_{i,j}} = q \pm \left(\O{\eps} - \frac{1}{\poly(n)}\right).\]
Then, \algref{alg:lp:approx:sampler:fast} reports $i$ with probability
\begin{align*}
\sum_{j \in\left[n^{c}\right]} &\PPr{\calE^*, \neg \mathrm{FAIL} \mid \calE_{i,j}} \PPr{\calE_{i,j}} \\
&=\sum_{j \in\left[n^{c}\right]} \frac{\left|x_i\right|^p}{\|X\|_p^p}\left(q \pm \left(\O{\eps} - \frac{1}{\poly(n)}\right)\right) \\ 
&=\frac{\left|x_i\right|^p}{\|x\|_p^p}\left(q \pm \left(\O{\eps} - \frac{1}{\poly(n)}\right)\right).
\end{align*}
Hence, given that the sampler reports some index, the probability of reporting $i \in [n]$ is
\[\frac{\left|x_i\right|^p}{\|x\|_p^p}(1 \pm \eps) \pm \frac{1}{\poly(n)},\]
which proves the correctness of our approximate sampler.
\end{proof}

\paragraph{Space complexity.} We analyze the space complexity for our approximate sampler. First, we bound the size of the set of large indices $B$ recovered by $\CountSketch_1$.

\begin{lemma}
\lemlab{lem:size:B}
Recall that set $B$ recovers the large indices in $\CountSketch_1$. We have $|B| = \polylog \frac{1}{\eps}$ with probability $1 - \frac{1}{\poly(n)}$.
\end{lemma}
\begin{proof}
By \lemref{lem:large:index}, there are at most $\polylog \frac{1}{\eps}$ indices $k \in [n]$ such that it satisfies, $\bv_k > \frac{n^{c/p}\|x\|_p}{400\log \frac{1}{\eps}}$. By \lemref{lem:cs:err:w:descrete}, the error of $\CountSketch_1$ table is at most $\frac{n^{c/p}\|x\|_p}{400\log \frac{1}{\eps}}$. Both events happen with high probability. Since we add an index $k$ to $B$ if the estimate of $\frac{n^{c/p}\|x\|_p}{200\log \frac{1}{\eps}}$, there are at most $\polylog \frac{1}{\eps}$ items in $B$. 
\end{proof}

With the bound on $|B|$, we can show the following space complexity.
\begin{lemma}
\lemlab{lem:space}
\algref{alg:lp:approx:sampler:fast} uses $\O{n^{1-2/p}\log^2 n \log \frac{1}{\eps}}$ bits of space.
\end{lemma}
\begin{proof}
$\CountSketch_1$ table has size $[n^{1-2/p}\log \frac{1}{\eps}] \times [\Theta(\log n)]$. By \lemref{lem:size:B}, we only need to maintain $\polylog \frac{1}{\eps}$ buckets in each row of $\CountSketch_2$. So, we need 
$$\O{n^{1-2/p}\log^2 n \log \frac{1}{\eps} + \log n \polylog \left(\frac{1}{\eps}\right)},$$
bits of space in total. Suppose that $\polylog \left(\frac{1}{\eps} \right) < n^{1-2/p} \log n$, the second term is dominated by the first term.
\end{proof}

\paragraph{Fast update sketch.} 
Next, we describe how to implement the discretization in a fast update time. We modify the fast-update sketch in \cite{JayaramW18} to fit in our $\CountSketch$ algorithms with random signs. Our goal is to compute the set of duplicated exponential variables $\{\rnd_\eta(1/\be_1^{1/p}), \ldots, \rnd_\eta(1/\be_{n^c}^{1/p})\}$ for each index $i \in [n]$. Note that the support size of $\rnd_\eta(x)$ for $x\in [\frac{1}{\poly(n)},\poly(n)]$ is $\O{\frac{1}{\eta}\log n}$, so we can count the number of inverse exponential variables that round up to each value in the support of $\rnd_\eta(x)$.

We define $I_q = (1+\eta)^q$ for $q \in \mathbb{Z} \cap [-Q,Q]$ where $Q=\O{\frac{1}{\eta}\log n}$. Let $\phi(x)$ be the cdf of the $1/p$-th power of the inverse exponential distribution. Then, for a standard exponential variable $\be$, the probability that $\rnd_\eta(1/\be^{1/p})$ equals $ I_q$ is $p_q = \phi(I_{q+1}) - \phi(I_q)$. The number $D_q$ of such inverse exponential variables follows a binomial distribution $\mathrm{Bin}(n^c, p_q)$.

Now, upon each arrival, there are $D_q$ updates with value $I_q$ that need to be hashed into $\CountSketch_2$. We can generate variables from multinomial distribution to compute $a_{i,j}^q$, which is the number of items in the $D_q$ updates that are hashed to bucket $(i,j)$ in the CountSketch table. Our next goal is to calculate $\sum_{t=1}^{a_{i,j}^q} g_t \cdot I_q$, which is the additive value to the estimator of bucket $(i,j)$. Note that $g_t$'s are Rademacher variables, so the sum follows the distribution $\mathrm{Bin}(a_{i,j}^q, 1/2) - a_{i,j}^q$. Therefore, it suffices to generate one binomial variable for each bucket to compute the sum. 

For $\CountSketch_1$, we only hash $\Theta(\log n)$ items upon each arrival. To make it consistent with $\CountSketch_2$, we find the smallest $q$ such that $D_q$ is not zero, and we use $I_q$ to simulate the maximum of the $n^c$ duplications. Then, for each item hashed to $\CountSketch_1$, we generate a geometric variable $g_p$ with parameter $p = \frac{1}{L}$, and we hash it to the bucket located $g_p$ positions after the bucket that receives the previous item.

Last, in our $L_2$ estimation, we use Gaussian variables to scale our vector instead of random signs. However, we can use a similar way to speed up our calculation. Consider $\sum_{t=1}^{a_{i,j}^q} \phi_t \cdot I_q$ where $\phi_t \sim \calN(0,1)$ are the Gaussian variables. Utilizing the $2$-stability of Gaussian variables, we have $\sum_{t=1}^{a_{i,j}^q} \phi_t \cdot I_q \sim g \sqrt{a_{i,j}^q} I_q$ where $g \sim \calN(0,1)$. 
Thus, it suffices to generate one Gaussian variable for each bucket to compute the sum. 

Next, we state the correctness of our fast-update sketch.
\begin{lemma}
Our fast update sketch results in the same distribution over the CountSketch table and the $L_2$-estimation scheme as the original algorithm. Upon each arrival in the stream, the update time is $\frac{1}{\eps} \polylog(n,\frac{1}{\eps})$.
\end{lemma}
\begin{proof}
In the $\CountSketch_1$ table, we hash each item to each bucket with probability $\frac{1}{L}$. The geometric variable with parameter $\frac{1}{L}$ characterizes the distribution of the number of buckets between two consecutive buckets that have the hashed item. Thus, our hashing scheme gives the same distribution.
For the $\CountSketch_2$ table, the multinomial variables give the correct hashing distribution by generating the number of items that are hashed to each bucket. Then, each bucket is increased by $\sum_{t=1}^{a_{i,j}^q} g_t \cdot I_q$ where $g_t$'s are Rademacher variables by our algorithm construction. Since $\sum_{t=1}^{a_{i,j}^q} g_t \cdot I_q$ has the same distribution as $\mathrm{Bin}(a_{i,j}^q, 1/2) - a_{i,j}^q$, our fast update sketch gives precisely the same distribution as the original two-stage CountSketch.
The correctness of the fast $L_2$ estimation follows from Lemma 6 in \cite{JayaramW18}. 

Now, we compute the update time. In $\CountSketch_1$, we generate $\O{\log n}$ geometric random variables upon each stream update. In $\CountSketch_2$, we have $\log n \polylog \left(\frac{1}{\eps}\right)$ buckets as specified in \lemref{lem:space}. For each bucket, it takes $\O{\log n}$ time to compute the additive value when an item arrives. Due to our choice of discretization factor $\eta = \frac{\eps}{\polylog \frac{1}{\eps}}$, the update time is $\frac{1}{\eps} \polylog(n,\frac{1}{\eps})$.
\end{proof}

\paragraph{Algorithm derandomization.} 
Note that the prior analysis relies on the assumption that independent exponential random variables and geometric random variables can be both generated and stored efficiently. 
We remark that since we are focused on achieving tight space bounds, we cannot afford to lose additional polylogarithmic factors from pseudorandom generators such as Nisan's PRG~\cite{Nisan92}. 

Instead, to derandomize our algorithm, we use the same approach as \cite{JayaramW18}, which we include here for completeness. 
The approach leverages the pseudorandom generator (PRG) developed by \cite{GopalanKM18}, which fools specific families of Fourier transforms, including, crucially for our setting, collections of half-space queries. 

Formally, for each $i \in [\lambda]$, a half-space query $H_i:\mathbb{R}^n \to \{0,1\}$ is defined on an input $Z=(z_1,\ldots,z_n)$ by the indicator function $\mathbf{1}\left[\alpha_1^{(i)}z_1+\ldots+\alpha_n^{(i)}z_n>\theta_i\right]$, where $\alpha^{(i)}\in\mathbb{Z}^n$ and $\theta_i\in\mathbb{Z}$ for all $i\in[n]$.

\begin{definition}[$\lambda$-half-space tester]
Given input $Z=(z_1,\ldots,z_n)$, a $\lambda$-half-space tester is a Boolean function $\sigma(H_1(Z),\dots,H_\lambda(Z))\in\{0,1\}$, where $\sigma:\{0,1\}^\lambda\to\{0,1\}$ and $H_1,\dots,H_\lambda$ are half-space queries. 
We say the tester is $M$-bounded if each input coordinate $z_i$ is drawn from a distribution over integers bounded in magnitude by $M$, and all coefficients $\alpha_j^{(i)}$ and thresholds $\theta_i$ also have magnitude at most $M$. 
\end{definition}

We now state a key result concerning the PRG for half-space testers, following Lemma 7 and Proposition 8 of \cite{JayaramW18}.

\begin{theorem}[\cite{GopalanKM18,JayaramW18}]
\thmlab{thm:gkm:prg}
Let $\calD$ be a distribution over $\{-M,\ldots,M\}$ that can be sampled using $\O{\log M}$ random bits. 
Let $Z=(z_1,\ldots,z_n)$ where each $z_i\sim \calD$ independently. 
Then for any $\eps>0$ and $c\ge1$, there exists $\ell=\O{\lambda\log\frac{nM}{\eps} (\log\log\frac{nM}{\eps})^2}$ and a deterministic function $F:\{0,1\}^\ell\to\{-M,\ldots,M\}^n$ such that for any $\lambda$-half-space tester $\sigma_H$, we have
\[\left\lvert\EEx{Z\sim\calD^n}{\sigma_H(Z)}- \EEx{y\sim\{0,1\}^\ell}{\sigma_H(F(y))}\right\rvert\le\eps.
\]
Furthermore, if $F(y)\in\{-M,\ldots,M\}^n$, then each coordinate of $F(y)$ can be computed using $\O{\ell}$ space and in $\polylog(nM)$ time.
\end{theorem}
We also require the following result:
\begin{lemma}
\lemlab{lem:derandom:many}
\cite{JayaramW18,WoodruffZ21}
Let $q \ge 1$ be a constant, and consider vectors $f_1, \ldots, f_q \in \mathbb{Z}^n$, each bounded by $M = \poly(n)$ and defined by updates from a stream $S$ within specified intervals $[t_{i,1}, t_{i,2}]$ for $i \in [q]$. 
Let $\calA$ be a streaming algorithm that maintains linear sketches $A \cdot f_1, \ldots, A \cdot f_q$, with $A \in \mathbb{R}^{k \times n}$ having i.i.d. entries that can be sampled with $\O{\log n}$ bits. 
Suppose that the output is given by $g(A \cdot f_1, \ldots, A \cdot f_q)$, where $g: \mathbb{R}^q \to \mathbb{R}$ is a composition function.

For any constant $c \ge 1$, $\calA$ can be implemented using a random matrix $A'$ with $\O{k \log n (\log \log n)^2}$ bits of space, satisfying:
\[|\PPr{g(A \cdot f_1, \ldots, A \cdot f_q) = y} - \PPr{g(A' \cdot f_1, \ldots, A' \cdot f_q) = y}| < n^{-ck},\]
for all $y \in \mathbb{R}^k$ with entry-wise bit complexity $\O{\log n}$.
\end{lemma}
We now turn to the task of derandomizing our algorithm. 
Recall that the algorithm relies on two sources of randomness: the hash functions and random signs in $\CountSketch$, and the exponential random variables. 
Following the approach of \cite{JayaramW18}, we apply two separate pseudorandom generators. 
The first instance simulates the random variables in $\CountSketch$ and the second instance simulates the exponential random variables. 
Let $R_c$ and $R_e$ respectively denote the random bits required for $\CountSketch$ and for the exponential variables. 
Since there are at most $\poly(n)$ random variables and each have magnitude $\poly(n)$, then at most $\poly(n)$ random bits are required for both $R_c$ and $R_e$. 

For a fixed index $i\in[n]$, we can view the behavior of the algorithm as a Boolean tester $\calA_i(R_c,R_e)\in\{0,1\}$, which evaluates to $1$ if the $L_p$ sampler selects coordinate $i$. 
Our goal is to show that we can replace both $R_c$ and $R_e$ with pseudorandom inputs generated by the PRG without significantly changing the output distribution. 
Specifically, we show there exist functions $F_1$ and $F_2$ such that
\[
\left\lvert\PPPr{R_c,R_e}{\calA_i(R_c,R_e)=1}-\PPr{y_1, y_2}{\calA_i(F_1(y_1),F_2(y_2))=1}\right\rvert\le \frac{1}{n^C},
\]
for any constant $C>0$, where $y_1$ and $y_2$ are seeds of length $\O{n^{1-2/p}\log^2 n\log\frac{1}{\eps}(\log\log n)^2}$ chosen appropriately for the desired level of accuracy.

We thus have the following guarantees of our algorithm:
\begin{theorem}
\thmlab{thm:approx:main}
Given a general turnstile stream $x$ and an accuracy parameter $\eps \in (0,1)$, there is an approximate sampler which outputs an index $i \in[n]$ with probability $\frac{|x_i|^p}{\|x\|_p^p}\cdot(1\pm \eps)$,
and outputs FAIL with probability at most $0.1$. The algorithm uses
$n^{1-2/p} \log^2 n \log \frac{1}{\eps}\cdot\poly(\log\log(n))$ bits of space to run. The update time is $\frac{1}{\eps} \cdot \polylog\left(n,\frac{1}{\eps}\right)$. In addition, it gives a $(1+\eps)$-estimation to the sampled item using extra $\frac{1}{\eps^2} n^{1-2/p} \log^2 n \log \frac{1}{\eps}\cdot\poly(\log\log(n))$ bits of space.
\end{theorem}
\begin{proof}
We derandomize the $\poly(n)$ random variables as follows. 
Let $R_c$ denote the randomness used to generate randomness for $\CountSketch$ and let $R_e$ denote the randomness used to generate the exponential variables. 
For any fixed index $i\in[n]$ and fixed exponential randomness $R_e$, define the tester $\calA_{i,R_e}(R_c)$ that determines whether index $i$ is selected by the $L_p$ sampler. 
Here, $R_e$ is fixed, while $R_c$ is the remaining source of randomness.

We show that $\calA_{i,R_c}(R_g)$ can be fooled using a pseudorandom generator with seed length $\O{n^{1-2/p}\log^2 n(\log\log n)^2}$. 
With $R_e$ fixed, the sketch can be written as $AMx$, where $A$ is the $\CountSketch$ and $Z$ is the matrix consisting of the exponential scalings of the duplications. 
Then we can apply \lemref{lem:derandom:many} to obtain a deterministic function $F_1$ that fools all such tests with a seed $y_1$ of length $\O{n^{1-2/p}\log^2 n\log\frac{1}{\eps}(\log\log n)^2}$. 

To derandomize the exponential variables, we fix a pseudorandom instantiation $F_1(y_1)$ and apply a hybrid argument. 
We define $\calB_{i,F_1(y_1)}(R_e)$ as the tester with $R_c$ fixed. 
With high probability, all exponentials lie within the range $\left[\frac{1}{\poly(n)},\poly(n)\right]$. 
Hence, the total randomness required is $\poly(n)$. 
By an observation of \cite{JayaramW18}, it suffices to apply \thmref{thm:gkm:prg} to obtain a second deterministic function $F_2$ with seed $y_2$ of length $\O{n^{1-2/p}\log^2 n\log\frac{1}{\eps}(\log\log n)^2}$ that fools all tests over exponential randomness. 
We note that the original algorithm processes each update in $\polylog(n)$ time. 
Since evaluating the GKM PRG also requires only $\polylog(n)$ time per variable, the update time is $\frac{1}{\eps} \cdot \polylog\left(n,\frac{1}{\eps}\right)$. 

Finally, we discuss how to retrieve an accurate approximation of the sampled item. 
We run a separate $\CountSketch$ with $\frac{1}{\eps^2} n^{1-2/p} \log \frac{1}{\eps}$ buckets and $\Theta(\log n)$ rows on the vector $\bv \in \mathbb{R}^n$, where $\bv_i = \max_{j \in [n^c]} |x_i| \cdot \rnd_\eta(1/ \be^{1/p})$. 
Following the analysis of \lemref{lem:cs:err:w:descrete}, this gives an estimate of each coordinate in $\bv$ with additive error $\eps \cdot \frac{n^{c/p}\|x\|_p}{400\log \frac{1}{\eps}}$. 
Then, since the sampled item must satisfy $|\bv_i| > \frac{n^{c/p}\|x\|_p}{400\log \frac{1}{\eps}}$ to pass the threshold of our first-stage table $\CountSketch_1$, the above additive error is smaller than $\eps \cdot |\bv_i|$, which implies a $(1+\eps)$-estimation.
\end{proof}

\section{\texorpdfstring{$L_p$}{Lp} Sampler Lower Bound}
In this section, we provide a sketching dimension lower bound for the $L_p$ sampler. 
The analysis is based on Section D.2 in \cite{GangulyW18}, which shows the sketching lower bound of the $F_p$-estimation protocol. 
Our core idea is to construct two distributions $\alpha$ and $\beta$, and for an arbitrary vector $x$, we can decide whether $x$ is drawn from $\alpha$ or $\beta$ with probability at least $0.6$ using our $L_p$ sampler. 
On the other hand, \cite{GangulyW18} lower bounds the total variation distance between $\alpha$ and $\beta$ under the image of a ``low'' dimensional sketch, which implies a sketching lower bound to distinguish whether an arbitrary vector is from $\alpha$ or $\beta$. 
Thus, this translates to a sketching dimension lower bound for the $L_p$ samplers, which can further be used to achieve bit complexity lower bounds using the techniques of \cite{GribelyukLWYZ25}. 

We first introduce the hard distributions.
\begin{definition}[Hard distributions]
\deflab{def:hd}
We define the first distribution $\alpha$ as $\calN(0, {\bf I}_n)$, which is the $n$-dimensional standard multi-variate Gaussian distribution. 
Let $e_i$ be the unit basis vector, where its $i$-th entry is $1$, and the other entries are $0$. 
Let $x \sim \calN(0, {\bf I}_n)$, let $C$ be a sufficiently large constant, let $E_{n} = \mathbb{E}_{x \sim \calN(0, {\bf I}_n)}[\|x\|_p]$, and let $i$ be sampled uniformly from $[n]$. 
Then, let $z = x+C E_{n-1}\cdot e_i$. 
We define the second distribution $\beta$ as the distribution of the vector $z$.
\end{definition}
The next statement lower bounds the sketching dimension that distinguishes the two distributions.
\begin{theorem}[c.f. Section D.2.4 in \cite{GangulyW18}]
\thmlab{thm:lb:hd}
Fix a matrix $S \in \mathbb{R}^{r \times n}$. Let $\alpha$ and $\beta$ be defined as in \defref{def:hd}. Given an arbitrary $x$ drawn from $\alpha$ or $\beta$, suppose that we can decide whether $x$ is from $\alpha$ or $\beta$ from $S\cdot x$ with probability at least $0.6$, then the sketching dimension $r$ is at least $\Omega\left(n^{1-2/p}\log n\right)$.
\end{theorem}

With the above result, we lower bound the sketching dimension of $L_p$ samplers by proposing a protocol that solves the problem in \thmref{thm:lb:hd} using $L_p$ samplers.

\begin{theorem}
\thmlab{thm:lb:sampler}
Let $x \in \mathbb{R}^n$ be a vector. Suppose that there is a linear sketch that outputs an index $i \in [n]$ with probability $\frac{|x_i|^p}{\|x\|_p^p}\cdot (1\pm 0.01)$, and outputs FAIL with probability at most $0.1$. Then, its sketching dimension is at least $\Omega\left(n^{1-2/p}\log n\right)$.
\end{theorem}
\begin{proof}
We claim that a linear sketch $S \in \mathbb{R}^{r\times n}$ that reports an approximate $L_p$ sampler with probability at least $0.9$ implies a method to distinguish $\alpha$ and $\beta$ in \thmref{thm:lb:hd} with probability at least $0.6$.
Given an arbitrary $x$ drawn from $\alpha$ or $\beta$, we take two $L_p$ samples from the vector $x$ using the linear sketch. 
We state that $x$ is from $\beta$ if both instances succeed and they sample the same coordinate.
Next, we show the correctness of this protocol.

First, we suppose that $x$ is from $\beta$, it suffices to sample the large coordinate $x_i = CE_{n-1}$.
From the standard properties of the multi-variate Gaussian distribution, we know $E_{n} = \Theta(n^{1/p})$ and $\mathbb{E}_{x \sim \calN(0, {\bf I}_n)}[\|x\|_p^p] = \Theta(n)$. 
Therefore, choosing a sufficiently large constant $C$ in \defref{def:hd}, we have $\frac{|x_i|^p}{\|x\|_p^p} \ge 0.99$ in expectation. 
Then, $i$ is sampled twice with probability at least $0.9$ conditioned on not failing. 
Thus, our protocol identifies $x$ from $\beta$ with probability at least $0.9$.
Next, if $x$ is from $\alpha$, we misclassify the case if some index is sampled twice.
For each coordinate, it is sampled twice with probability $\frac{1}{n^2}$, so by a union bound, the probability of misclassification is at most $\frac{1}{n}$.
Thus, for a large enough $n$, we decide whether $x$ is from $\alpha$ or $\beta$ with probability at least $0.6$. 

Then, we have $r = \Omega(n^{1-2/p}\log n)$ from the sketching lower bound in \thmref{thm:lb:hd}. By Yao's minimax lemma, we have our desired result.
\end{proof}

\section{Additional Applications}
In this section, we describe a number of additional applications of our techniques. 
In \secref{sec:norm:estimate}, we first give applications to norm estimation for a specific subset of coordinates whose identity is only revealed after the data stream is processed. 
We then describe a number of other perfect samplers in \secref{sec:other:samplers}. 
Last, we provide a framework for obtaining $G$-samplers for general functions using rejection sampling in \secref{sec:framework}.

\subsection{Application to Norm Estimation}
\seclab{sec:norm:estimate}
To estimate $\|x_{\calQ}\|_p^p=\sum_{i\in\calQ}x_i^p$, we first use our $L_p$ sampler to produce an index $i\in[n]$. 
We also produce an unbiased estimate $\widehat{F_p}$ to $\|x\|_p^p$ using existing techniques~\cite{Ganguly15}. 
Now if $i\in\calQ$, which is revealed on the query $\calQ$, then we set the estimate $Y=\widehat{F_p}$, and otherwise, we set $Y=0$. 
Observe that $\Ex{Y}=\|x_{\calQ}\|_p^p\pm\frac{1}{\poly(n)}$ and moreover the variance is at most $\|x_{\calQ}\|_p^p\cdot\|x\|_p^p$. 
Thus to drive the variance small enough to obtain a $(1+\eps)$-approximation, it suffices to repeat $\O{\frac{1}{\alpha\eps^2}}$ times, since $\|x_{\calQ}\|_p^p\ge\alpha\|x\|_p^p$ for $\alpha\in(0,1)$. 
We give the full algorithm in \algref{alg:fp:query}. 

\begin{algorithm}[!htb]
\caption{Moment estimation for a query subset $\calQ$}
\alglab{alg:fp:query}
\begin{algorithmic}[1]
\Require{Vector $x\in \mathbb{R}^n$ defined by a turnstile stream, accuracy parameter $\eps\in(0,1)$, post-processing set $\calQ\subseteq[n]$, ratio parameter $\alpha\in(0,1]$}
\Ensure{A $(1+\eps)$-estimation to $\|x\|_p^p$}
\State{$R\gets\O{\frac{1}{\alpha\eps^2}}$}
\For{$r\in[R]$}
\State{Approximate $L_p$ sample with distortion $\frac{\eps}{4}$ an index $i_r\in[n]$}
\State{Let $C_r$ be an estimate of $\|S_r\|_p^p$}
\Comment{Ganguly's estimator, see \thmref{thm:ganguly}}
\EndFor
\State{At the end of the stream, process $\calQ$}
\State{\Return $Z=\frac{1}{R}\sum_{r: i_r\in\calQ}C_r$}
\end{algorithmic}
\end{algorithm}
Now, we present the following result from \cite{Ganguly15} which produces unbiased $F_p$ estimates using optimal space.
\begin{theorem}[\cite{Ganguly15}]
\thmlab{thm:ganguly}
For $p>2$ and $0<\eps \leq 1$, there exists an algorithm in the turnstile stream model that returns an estimator $\widehat{F_p}$ such that $\Ex{\widehat{F_p}} = F_p$ and $\Var{\widehat{F_p}} \le \frac{F_p^2}{50}$ with probability at least $0.9$. 
The algorithm uses $\O{n^{1-2 / p} \log^2 n}$ bits of space.
\end{theorem}
The next lemma proves that \algref{alg:fp:query} gives an $(1+\eps)$-estimation.
\begin{lemma}
\lemlab{lem:fp:query:correct}
For the output $Z$ of \algref{alg:fp:query},
\[\PPr{\left\lvert Z-\|x_{\calQ}\|_p^p\right\rvert\le\eps\cdot\|x_{\calQ}\|_p^p}\ge0.99.\]
\end{lemma}
\begin{proof}
For each $r\in[R]$, we define $Z_r=C_r$ if $i_r\in\calQ$ and $Z_r=0$ if $i_r\notin\calQ$, so that $Z=\frac{1}{R}\sum_{r\in[R]}Z_r$. 
We have
\[\Ex{Z_r}=\sum_{i\in\calQ}\left(\left(1\pm\eps\right)\frac{|x_i|^p}{\|x\|_p^p}+\frac{1}{\poly(n)}\right)\cdot\Ex{C_r}.\]
By \thmref{thm:ganguly}, $\Ex{C_r}=\|x\|_p^p$, so that
\[\Ex{Z_r}=\left(1\pm\frac{\eps}{4}\right)\sum_{i\in\calQ} |x_i|^p=\left(1\pm\frac{\eps}{4}\right)\|x_{\calQ}\|_p^p.\]
Moreover, we have
\[\Ex{Z_r^2}=\sum_{i\in\calQ}\left(\left(1\pm\frac{\eps}{4}\right)\frac{|x_i|^p}{\|x\|_p^p}+\frac{1}{\poly(n)}\right)\cdot\Ex{C_r^2}.\]
By \thmref{thm:ganguly}, we have $\Ex{C_r^2}=\frac{\|x\|_p^p}{50}$, so that
\[\Ex{Z_r}\le\frac{1}{25}\sum_{i\in\calQ} |x_i|^p\cdot\|x\|_p^p=\frac{1}{25}\|x_{\calQ}\|_p^p\cdot\|x\|_p^p.\]
Therefore, we have $\Ex{Z}=\left(1\pm\frac{\eps}{4}\right)\|x_{\calQ}\|_p^p$ and
\[\Var{Z}\le\frac{1}{R^2}\sum_{r\in[R]}\frac{1}{25}\|x_{\calQ}\|_p^p\cdot\|x\|_p^p=\frac{1}{25R}\|x_{\calQ}\|_p^p\cdot\|x\|_p^p.\]
Since $R=\O{\frac{1}{\alpha\eps^2}}$ and $\|x_{\calQ}\|_p^p\ge\alpha\|x\|_p^p$ by assumption, then we have
\[\Var{Z}\le\O{\eps^2}\cdot\|x_{\calQ}\|_p^{2p}.\]
By Chebyshev's inequality, it follows that
\[\PPr{\left\lvert Z-\|x_{\calQ}\|_p^p\right\rvert\le\eps\cdot\|x_{\calQ}\|_p^p}\ge0.99.\]
\end{proof}
Now, we state the formal theorem of norm estimation for a query subset.
\begin{theorem}
Given $p>2$, there exists an algorithm that processes a turnstile stream defining a vector $x\in\mathbb{R}^n$ and a post-processing query set $\calQ\subseteq[n]$, and with probability at least $0.99$, outputs a $(1+\eps)$-approximation to $\|x_{\calQ}\|_p^p$. 
For $\|x_{\calQ}\|_p^p\ge\alpha\|x\|_p^p$, the algorithm uses $\tO{\frac{1}{\alpha\eps^2}n^{1-2/p}}$ bits of space. 
\end{theorem}
\begin{proof}
Consider \algref{alg:fp:query}. 
The justification of correctness of approximation results from \lemref{lem:fp:query:correct}. 
Thus it remains to analyze the space complexity of \algref{alg:fp:query}. 
We can use either a perfect $L_p$ sampler or an approximate $L_p$ sampler to acquire the samples $i_r$. 
By \thmref{thm:lp_sampler_frac} or \thmref{thm:approx:main}, these subroutines use $\tO{n^{1-2/p}}$ bits of space. 
By \thmref{thm:ganguly}, the $F_p$ estimation algorithm for producing $C_r$ also uses $\tO{n^{1-2/p}}$ bits of space. 
Since each subroutine is repeated $R=\O{\frac{1}{\alpha\eps^2}}$ times, then the total space complexity is $\tO{\frac{1}{\alpha\eps^2}n^{1-2/p}}$ bits of space, as claimed. 
\end{proof}

\subsection{Additional Samplers}
\seclab{sec:other:samplers}
In this section, we give perfect $G$-samplers for the logarithmic function $G(z)=\log(1+|z|)$ and the cap function $G(z)=\min(T,|z|^p)$. 

For both samplers, we require the following perfect $L_0$-sampler, which outputs a coordinate $i\in[n]$ with probability $\frac{1}{\|x\|_0}+\frac{1}{\poly(n)}$, along with the exact value of $x_i$. 
\begin{theorem}
\thmlab{thm:perfect:lzero}
\cite{JowhariST11}
There exists a perfect $L_0$ sampler on turnstile streams that succeeds with probability at least $1-\delta$, using $\O{\log^2 n\log\frac{1}{\delta}}$ bits of space. 
Moreover, for the index $i\in[n]$ that it returns, it returns $x_i$ exactly. 
\end{theorem}
Note that if we use \thmref{thm:perfect:lzero} to sample an index $i\in[n]$ with probability $\frac{1}{\|x\|_0}+\frac{1}{\poly(n)}$ and obtain $x_i$ exactly, then we can output $i$ with probability $\frac{G(x_i)}{H}$, where $H$ is some fixed quantity that upper bounds $G(x_i)$ for all $i\in[n]$, to guarantee that the probability is well-defined. 
For the logarithmic function $G(z)=\log(1+|z|)$, we can choose $H=\log(1+m)$ and for the cap function $G(z)=\min(T,|z|^p)$, we choose $H=T$. 
Thus it remains to acquire a sufficiently large number of independent $L_0$ samples to accept some sample with probability. 
Since a sample is accepted with probability at least $\frac{\log 2}{\log(m+1)}$ for the logarithmic function, we can use $\O{\log m}$ independent $L_0$ samples. 
Similarly, a sample is accepted with probability at least $\frac{1}{T}$ for the cap function, so it suffices to use $\O{T}$ independent $L_0$ samples. 
\begin{algorithm}[!htb]
\caption{Perfect $G$-sampler for $G(z)=\log(1+|z|)$}
\alglab{alg:log:sample}
\begin{algorithmic}[1]
\Require{Vector $x \in \mathbb{R}^n$ defined by a turnstile stream}
\Ensure{$G$-sample from $x$ for $G(z)=\log(1+|z|)$}
\State{Let $m$ be the length of the stream, $R\gets\O{\log m}$}
\For{$r\in[R]$}
\State{Acquire a perfect $L_0$ sample $i_r$ with failure probability $0.01$}
\State{\Return $i_r$ and abort, with probability $\frac{\log(1+|x_{i_r}|)}{\log m}$}
\EndFor
\end{algorithmic}
\end{algorithm}

The next theorem states the result of $G$-sampler for the logarithmic function.
\begin{theorem}
For the function $G(z)=\log(1+|z|)$, there exists a $G$-sampler on turnstile streams that uses $\O{\log^3 n}$ bits of space and succeeds with probability $0.99$. 
\end{theorem}
\begin{proof}
Consider \algref{alg:log:sample}. 
Let $N=\|x\|_0$ be the number of nonzero coordinates and let $\calS$ be the set of nonzero coordinates, so that $N=|\calS|$. 
Then for each $i\in\calS$ and each fixed $r\in[R]$, the algorithm samples $i$ with probability $\frac{1}{N}+\frac{1}{\poly(n)}$ and then chooses to accept $i$ with probability $\frac{\log(1+|x_i|)}{\log m}$. 
Therefore, the probability the sample returns $i$ is $\frac{\log(1+|x_i|)}{N\log m}+\frac{1}{\poly(n)}$. 
Hence, conditioned on some index being returned, the probability $i$ is sampled is
\[\frac{\log(1+|x_i|)}{\sum_{j\in[n]}\log(1+|x_j|)}+\frac{1}{\poly(n)},\]
as desired. 

On the other hand, for each $r\in[R]$, the sample $i_r$ is returned with probability $\frac{\log(1+|x_{i_r}|)}{\log m}\ge\frac{\log 2}{\log m}$. 
For sufficiently large $R=\O{\log m}$, at least $\O{\log m}$ instances of the perfect $L_0$ sampler will succeed. 

Hence, the algorithm outputs a sample with probability at least $0.99$ for $R=\O{\log m}$.  
By \thmref{thm:perfect:lzero}, each $L_0$ sampler with failure probability $0.01$ uses space $\O{\log^2 n}$. 
Therefore, for $\log m=\O{\log n}$, \algref{alg:log:sample} uses $\O{\log^3 n}$ bits of space in total. 
\end{proof}

We now consider a perfect $G$-sampler for the function $G(z)=\min(T,|z|^p)$ for a threshold $T$, and for any $p\ge 0$ by rejection sampling. 

\begin{algorithm}[!htb]
\caption{Perfect $G$-sampler for $G(z)=\min(T,|z|^p)$}
\alglab{alg:cap:sample}
\begin{algorithmic}[1]
\Require{Vector $f \in \mathbb{R}^n$ arrived in a data stream}
\Ensure{$G$-sample from $f$ for $G(z)=\min(T,|z|^p)$}
\State{Let $m$ be the length of the stream, $R\gets\O{T}$}
\For{$r\in[R]$}
\State{Acquire a perfect $L_0$ sample $i_r$ with failure probability $0.01$}
\State{\Return $i_r$ and abort, with probability $\frac{\min(T,|x_{i_r}|^p)}{T}$}
\EndFor
\end{algorithmic}
\end{algorithm}

The next theorem states the result of $G$-sampler for the cap function.
\begin{theorem}
For the function $G(z)=\min(T,|z|^p)$, there exists a $G$-sampler on turnstile streams that uses $\O{T\log^2 n}$ bits of space and succeeds with probability $0.99$. 
\end{theorem}
\begin{proof}
Consider \algref{alg:cap:sample}. 
We define $N=\|x\|_0$ to be the number of nonzero coordinates and we define $\calS$ to be the set of nonzero coordinates, so that $N=|\calS|$. 
For each $i\in\calS$ and each fixed $r\in[R]$, the algorithm samples $i$ with probability $\frac{1}{N}+\frac{1}{\poly(n)}$ and then outputs $i_r$ with probability $\frac{\min(T,|x_{i_r}|^p)}{T}$. 
Hence, the probability the sample outputs $i$ is $\frac{\min(T,|x_{i_r}|^p)}{TN}+\frac{1}{\poly(n)}$. 
Therefore, conditioned on the algorithm outputting a sample, the probability that index $i$ is sampled is
\[\frac{\min(T,|x_i|^p)}{\sum_{j\in[n]}\min(T,|x_j|^p)}+\frac{1}{\poly(n)},\]
which is the desired sampling probability for $G(z)=\min(T,|z|^p)$. 

Moreover, for $r\in[R]$, the sample $i_r$ is output with probability $\frac{\min(T,|x_{i_r}|^p)}{T}\ge\frac{1}{T}$. 
For sufficiently large $R=\O{T}$, at least $\O{T}$ instances of the perfect $L_0$ sampler will succeed. 
Thus, the algorithm will successfully output a sample with probability at least $0.99$ for $R=\O{T}$.  

By \thmref{thm:perfect:lzero}, each perfect $L_0$ sampler with failure probability $0.01$ uses space $\O{\log^2 n}$. 
Therefore, for $\log m=\O{\log n}$, \algref{alg:log:sample} uses $\O{T\log^2 n}$ bits of space in total. 
\end{proof}

\subsection{Rejection Sampling Framework}
\seclab{sec:framework}
In this section, we generalize the rejection sampling method in \secref{sec:other:samplers} to provide a framework that gives perfect $G$-samplers for a general function $G(z) \ge 0$. Suppose that we have access to an upper bound $H$ to $\max_{z\in [m]} G(z)$ and a lower bound $Q$ to $\min_{z\in [m]} G(z)$. First, we use \thmref{thm:perfect:lzero} to acquire $L_0$ samples and obtain the sampled item $x_i$ exactly. Next, we output $i$ with probability $\frac{G(x_i)}{H}$. The rejection probability is well-defined since $G(z) \le H$ for all $z$, in addition, we sample from the correct distribution due to the choice of the rejection probability. Note that a $L_0$ sample is accepted with probability at least $\frac{Q}{H}$, so it suffices to use $\O{\frac{H}{Q}}$ independent $L_0$ samples.
\begin{algorithm}[!htb]
\caption{Perfect $G$-sampler for $G(z)$}
\alglab{alg:slow:sample}
\begin{algorithmic}[1]
\Require{Vector $x \in \mathbb{R}^n$ defined by a turnstile stream, function $G(z)$, parameters $H \ge \max_{z\in[m]} G(z)$, $Q\le \min_{z\in[m]} G(z)$}
\Ensure{$G$-sample from $x$ for $G(z)$}
\State{Let $m$ be the length of the stream, $R\gets\O{\frac{H}{Q}}$}
\For{$r\in[R]$}
\State{Acquire a perfect $L_0$ sample $i_r$ with failure probability $0.01$}
\State{\Return $i_r$ and abort, with probability $\frac{G(i_r)}{H}$}
\EndFor
\end{algorithmic}
\end{algorithm}

Now, we give our formal theorem statement for the rejection sampling framework.
\begin{theorem}
For the function $G(z)$ satisfying $Q\le G(z)\le H$ for all $z\in[m]$, there exists a $G$-sampler on turnstile streams that uses $\O{\frac{H}{Q}\log^2 n}$ bits of space and succeeds with probability $0.99$. 
\end{theorem}
\begin{proof}
Consider \algref{alg:slow:sample}. 
Let $N=\|x\|_0$ be the number of nonzero coordinates and let $\calS$ be the set of nonzero coordinates, so that $N=|\calS|$. 
Then for each $i\in\calS$ and each fixed $r\in[R]$, the algorithm samples $i$ with probability $\frac{1}{N}+\frac{1}{\poly(n)}$ and then chooses to accept $i$ with probability $\frac{G(x_i)}{H}$. 
Therefore, the probability the sample returns $i$ is $\frac{G(x_i)}{N H}+\frac{1}{\poly(n)}$. 
Hence, conditioned on some index being returned, the probability $i$ is sampled is
\[\frac{G(x_i)}{\sum_{j\in[n]}G(x_j)}+\frac{1}{\poly(n)},\]
as desired. 

On the other hand, for each $r\in[R]$, the sample $i_r$ is returned with probability $\frac{G(x_{i_r})}{H}\ge\frac{Q}{H}$. 
Hence, the algorithm outputs a sample with probability at least $0.99$ for $R=\O{\frac{H}{Q}}$.  
By \thmref{thm:perfect:lzero}, each $L_0$ sampler with failure probability $0.01$ uses space $\O{\log^2 n}$. 
Therefore, \algref{alg:slow:sample} uses $\O{\frac{H}{Q}\cdot\log^2 n}$ bits of space in total. 
\end{proof}

\section*{Acknowledgments}
David P. Woodruff is supported in part by Office of Naval Research award number N000142112647 and a Simons Investigator Award.
Shenghao Xie and Samson Zhou were supported in part by NSF CCF-2335411.

\def\shortbib{0}
\bibliographystyle{alpha}
\bibliography{references}

\end{document}